\PassOptionsToPackage{unicode}{hyperref}
\PassOptionsToPackage{hyphens}{url}
\PassOptionsToPackage{dvipsnames,svgnames,x11names}{xcolor}
\documentclass[
  11pt]{article}
\usepackage{xcolor}
\usepackage[margin=1.1in]{geometry}
\usepackage{amsmath,amssymb}
\usepackage{iftex}
\ifPDFTeX
  \usepackage[T1]{fontenc}
  \usepackage[utf8]{inputenc}
  \usepackage{textcomp} % provide euro and other symbols
\else % if luatex or xetex
  \usepackage{unicode-math} % this also loads fontspec
  \defaultfontfeatures{Scale=MatchLowercase}
  \defaultfontfeatures[\rmfamily]{Ligatures=TeX,Scale=1}
\fi
\usepackage{lmodern}
\ifPDFTeX\else
\fi
\IfFileExists{upquote.sty}{\usepackage{upquote}}{}
\IfFileExists{microtype.sty}{% use microtype if available
  \usepackage[]{microtype}
  \UseMicrotypeSet[protrusion]{basicmath} % disable protrusion for tt fonts
}{}
\makeatletter
\@ifundefined{KOMAClassName}{% if non-KOMA class
  \IfFileExists{parskip.sty}{%
    \usepackage{parskip}
  }{% else
    \setlength{\parindent}{0pt}
    \setlength{\parskip}{6pt plus 2pt minus 1pt}}
}{% if KOMA class
  \KOMAoptions{parskip=half}}
\makeatother
\makeatletter
\ifx\paragraph\undefined\else
  \let\oldparagraph\paragraph
  \renewcommand{\paragraph}{
    \@ifstar
      \xxxParagraphStar
      \xxxParagraphNoStar
  }
  \newcommand{\xxxParagraphStar}[1]{\oldparagraph*{#1}\mbox{}}
  \newcommand{\xxxParagraphNoStar}[1]{\oldparagraph{#1}\mbox{}}
\fi
\ifx\subparagraph\undefined\else
  \let\oldsubparagraph\subparagraph
  \renewcommand{\subparagraph}{
    \@ifstar
      \xxxSubParagraphStar
      \xxxSubParagraphNoStar
  }
  \newcommand{\xxxSubParagraphStar}[1]{\oldsubparagraph*{#1}\mbox{}}
  \newcommand{\xxxSubParagraphNoStar}[1]{\oldsubparagraph{#1}\mbox{}}
\fi
\makeatother

\usepackage{longtable,booktabs,array}
\usepackage{calc} % for calculating minipage widths
\usepackage{etoolbox}
\makeatletter
\patchcmd\longtable{\par}{\if@noskipsec\mbox{}\fi\par}{}{}
\makeatother
\IfFileExists{footnotehyper.sty}{\usepackage{footnotehyper}}{\usepackage{footnote}}
\makesavenoteenv{longtable}
\usepackage{graphicx}
\makeatletter
\newsavebox\pandoc@box
\newcommand*\pandocbounded[1]{% scales image to fit in text height/width
  \sbox\pandoc@box{#1}%
  \Gscale@div\@tempa{\textheight}{\dimexpr\ht\pandoc@box+\dp\pandoc@box\relax}%
  \Gscale@div\@tempb{\linewidth}{\wd\pandoc@box}%
  \ifdim\@tempb\p@<\@tempa\p@\let\@tempa\@tempb\fi% select the smaller of both
  \ifdim\@tempa\p@<\p@\scalebox{\@tempa}{\usebox\pandoc@box}%
  \else\usebox{\pandoc@box}%
  \fi%
}
\def\fps@figure{htbp}
\makeatother

\usepackage[]{natbib}
\usepackage{amsmath,amssymb,bm}
\usepackage{authblk}

\author[1,2]{Michael Bücker}
\affil[1]{Münster School of Business, FH Münster University of Applied Sciences}
\affil[2]{Institute for Process Management and Digital Transformation (IPD), FH Münster University of Applied Sciences}
\affil[ ]{Correspondence: \texttt{michael.buecker@fh-muenster.de}}
\makeatletter\let\author\@gobble\makeatother
\makeatletter
\@ifpackageloaded{caption}{}{\usepackage{caption}}
\AtBeginDocument{%
\ifdefined\contentsname
  \renewcommand*\contentsname{Table of contents}
\else
  \newcommand\contentsname{Table of contents}
\fi
\ifdefined\listfigurename
  \renewcommand*\listfigurename{List of Figures}
\else
  \newcommand\listfigurename{List of Figures}
\fi
\ifdefined\listtablename
  \renewcommand*\listtablename{List of Tables}
\else
  \newcommand\listtablename{List of Tables}
\fi
\ifdefined\figurename
  \renewcommand*\figurename{Figure}
\else
  \newcommand\figurename{Figure}
\fi
\ifdefined\tablename
  \renewcommand*\tablename{Table}
\else
  \newcommand\tablename{Table}
\fi
}
\@ifpackageloaded{float}{}{\usepackage{float}}
\floatstyle{ruled}
\@ifundefined{c@chapter}{\newfloat{codelisting}{h}{lop}}{\newfloat{codelisting}{h}{lop}[chapter]}
\floatname{codelisting}{Listing}

\usepackage{amsthm}
\theoremstyle{plain}
\newtheorem{proposition}{Proposition}[section]
\theoremstyle{plain}
\newtheorem{lemma}{Lemma}[section]
\theoremstyle{definition}
\newtheorem{definition}{Condition}[section]
\theoremstyle{plain}
\newtheorem{theorem}{Theorem}[section]
\theoremstyle{plain}
\newtheorem{corollary}{Corollary}[section]
\theoremstyle{remark}
\AtBeginDocument{}

\makeatother
\makeatletter
\@ifpackageloaded{caption}{}{\usepackage{caption}}
\@ifpackageloaded{subcaption}{}{\usepackage{subcaption}}
\makeatother
\usepackage{bookmark}
\IfFileExists{xurl.sty}{\usepackage{xurl}}{} % add URL line breaks if available
\hypersetup{
  pdftitle={The price of a weak shadow variable: inference under nonignorable missingness as relevance vanishes},
  pdfkeywords={Effective sample size, Exponential
tilt, Identification-robust inference, Missing not at
random, Nonresponse instrument},
  colorlinks=true,
  linkcolor={blue},
  filecolor={Maroon},
  citecolor={Blue},
  urlcolor={Blue},
  pdfcreator={LaTeX via pandoc}}

\title{The price of a weak shadow variable: inference under nonignorable
missingness as relevance vanishes}
\author{}
\date{2026-11-09}
\begin{document}
\maketitle
\begin{abstract}
When an outcome is missing not at random, identification is commonly
restored by a shadow variable: a covariate that predicts the outcome but
is excluded from the response mechanism. The literature treats the first
requirement, relevance, as a yes-or-no condition. We make it continuous,
parametrizing it by a scalar that scales the shadow variable's entire
contribution to the moment conditions, and ask what inference costs as
that scalar vanishes. One functional of the inverse weight governs the
answer. We characterize it exactly, show that it fixes the effective
sample size, and use it to sharpen the rates of the estimator that
avoids the weight. The characterization is concrete. Whether the
functional saturates or diverges turns not on boundedness of the outcome
but on whether the tilting variable has an atom at its essential
supremum: a binary outcome saturates, a uniform one does not. It also
bounds how far a scanned set reaches, and that reach grows only as the
square root of the logarithm of the sample size, so a thousandfold
increase buys one unit of tilt while the set to be scanned widens. A
reanalyzed credit-scoring study shows the cost: a characteristic an
optimizing selector uses cannot satisfy exclusion.

\textbf{Keywords:} Effective sample size · Exponential tilt ·
Identification-robust inference · Missing not at random · Nonresponse
instrument.
\end{abstract}

\section{Introduction}\label{sec-intro}

Let \(Y\) be an outcome observed only when \(R=1\), with covariates
\(X\) always observed, and suppose the response probability depends on
\(Y\) itself. Such missingness is not ignorable, and without further
structure neither the distribution of \(Y\) nor the parameters of a
model for it are identified. The device that has become standard is a
\emph{shadow variable}: a further covariate \(Z\), always observed, that
predicts \(Y\) but is conditionally independent of \(R\) given \((Y,X)\)
\citep{TangLittleRaghunathan2003, dHaultfoeuille2010, WangShaoKim2014, MiaoTchetgen2016}.
Under completeness conditions the model is then point identified, and a
range of estimators is available: semiparametric likelihood
\citep{QinLeungShao2002, ChenFang2019}, pseudo-likelihood
\citep{ZhaoShao2015}, and estimators that avoid the inverse weight
altogether \citep{LiMiaoTchetgen2022}.

In all of this the shadow variable's predictive power enters as a
condition to be assumed: completeness, or non-degeneracy of a
coefficient, or a rank condition. It is either satisfied or it is not.
The point is clearest where the shadow variable is not given but
searched for. \citet{ChenShaoFang2021}, \citet{ZhaoWangShao2021} and
\citet{ChenShao2025} select one from a set of candidate covariates, and
what their criteria detect is the exclusion restriction; relevance
enters as the companion requirement that the candidate be \emph{related}
to the outcome, which a candidate either is or is not.
\citet{ChenShaoFang2021} are explicit about the consequence: their
criterion takes the same limiting value zero at a candidate that is
excluded and relevant and at one that is excluded and irrelevant, so
what it consistently recovers is the set of covariates satisfying
exclusion, and relevance must be assumed away by hypothesis. One recent
line drops that condition rather than grading it:
\citet{ChenSimchiLeviXiong2026} call a shadow variable \emph{weak} when
completeness is not imposed, and obtain partial identification with
sharp bounds. The weakness studied here is different in kind. The
conditions hold and the parameter stays point identified for every
positive relevance; what vanishes is not identification but the
information the data carry about it. Yet in applications the condition
is satisfied by margins that vary enormously, and the question of what
happens in between is not addressed. The purpose of this paper is to
make that question precise and answer it.

That the question is worth asking is not a new observation.
\citet{NakanoHoshino2026} study a scalar mean under a linear shadow
model and obtain confidence sets valid along sequences in which the
relevance drifts to zero: as relevance vanishes their sets come out
unbounded, which is what a valid set must do once the data cease to
restrict the parameter. \citet{Dufour1997} had already shown, in
generality, that any valid confidence set for a parameter that is
unidentified on part of the parameter space must be unbounded with
probability close to its nominal level, and that Wald sets built from
asymptotic standard errors have coverage zero. The observation that
relevance may be weak is due to \citet{NakanoHoshino2026}. What we add
is an account of what weak relevance does to each component of the
machinery: to the weights, to the conditioning of the sieve, to the
feasibility of empirical likelihood, and to the arithmetic.

Three findings organize the paper. First, in the moment system that
identifies the model the entire contribution of the shadow variable
factorizes, at every point of the identified set, as a fixed function
multiplied by a scalar \(\rho\) (Section~\ref{sec-curve}); relevance is
therefore a global rather than a local property, and everything can be
followed as a function of it. Second, the second moment of the inverse
weight decomposes exactly, with all dependence on the response model
carried by a \(\chi^{2}\)-divergence \(\Lambda\) of an exponential tilt
(Section~\ref{sec-tilt}); an integral representation of \(\log\Lambda\)
yields a sharp dichotomy: \(\Lambda\) is capped when the tilting
variable has an atom at its essential supremum, and diverges along the
identified set when it does not. Boundedness of the outcome is not the
criterion---a uniform outcome is bounded and still diverges---though a
binary outcome, having an atom, falls on the capped side. Third,
\(\Lambda\) determines the effective sample size exactly, and thereby
how far a scanned confidence set can actually be computed. That reach
grows like \((\log n)^{1/2}\) and is free of \(\rho\), whereas the set
that must be scanned widens without limit as \(\rho\) falls
(Section~\ref{sec-reach}). The procedure prescribed to reveal
nonidentification therefore conceals it, and gives no sign of having
done so.

The remaining sections develop the estimation-side consequences. Working
with the representer construction of \citet{LiMiaoTchetgen2022}, which
avoids the inverse weight, Section~\ref{sec-representer} shows that a
cancellation peculiar to the shadow-variable construction makes several
rates one power of \(\rho\) smaller than norm bounds suggest, and that
the resulting threshold for the regularization coincides with the
independent condition for the design matrix to be well conditioned---a
coherence that the cruder bound does not have.
Section~\ref{sec-diagnostics} separates two failures that are easily
confused: feasibility of the empirical likelihood problem is governed by
the effective sample size and is flat in the strength of identification,
across a tenfold range of the latter. Section~\ref{sec-simulation}
reports a simulation study, and Section~\ref{sec-application} reanalyzes
a published credit-scoring study, reproducing it exactly before
diagnosing it, and shows that in that application---and in a second,
much larger one---no usable shadow variable existed.
Section~\ref{sec-selection} explains why that had to be so whenever the
response is generated by an agent scoring the same characteristics for
the same outcome, and where a valid shadow variable can still be sought.
One further term, \emph{exponential tilting}, names a different object
in \citet{Caner2010}; Section~\ref{sec-terminology} distinguishes the
two.

\section{Model, identification, and a relevance
parameter}\label{sec-curve}

\subsection{The model and its moment conditions}\label{sec-moments}

Let \(\beta_0\) solve \(E\{\psi(Y,X;\beta)\}=0\) for a known estimating
function \(\psi\), and suppose \(Y\) is observed only when \(R=1\), so
that the observation is \(O=(X,Z,R,RY)\), with \(X\) a vector of
covariates, \(Z\) the shadow variable and \(R\) the response indicator.
Write the working propensity as
\begin{equation}\phantomsection\label{eq-propensity}{
  \pi_\gamma(y,x)=\operatorname{expit}(\gamma_0+\gamma_X^{\mathrm{T}}x+\gamma_Y y),
}\end{equation} where \(\operatorname{expit}(u)=(1+e^{-u})^{-1}\) is the
inverse of the logit and \(\gamma=(\gamma_0,\gamma_X,\gamma_Y)\)
collects its coefficients. Write \(\pi_0\) for the propensity at the
true value of \(\gamma\)---the subscript marks the truth, not the
intercept coordinate \(\gamma_0\)---so that
\(\operatorname{pr}(R=1\mid Y,X,Z)=\pi_0(Y,X)\): the response
probability depends on the outcome itself, and not on the shadow
variable. The exclusion restriction defining a shadow variable is that
\(Z\) does not enter the response mechanism,
\begin{equation}\phantomsection\label{eq-exclusion}{
  Z\perp\!\!\!\perp R\mid(Y,X).
}\end{equation} As a conditional independence, (\ref{eq-exclusion})
imposes infinitely many restrictions. To estimate from it, one selects
finitely many test functions \(h(Z,X)\) and requires the weighting
residual \(R/\pi_\gamma-1\) to be uncorrelated with each of them, which
yields the stacked moment condition
\begin{equation}\phantomsection\label{eq-moments}{
  E\{g(O;\theta)\}=0,\qquad
  g(O;\theta)=\begin{pmatrix}
    (R/\pi_\gamma)\,\psi(Y,X;\beta)\\[2pt]
    (R/\pi_\gamma-1)\,h(Z,X)
  \end{pmatrix},
  \qquad \theta=(\beta,\gamma).
}\end{equation} The moment system (\ref{eq-moments}) is the construction
of \citet{QinLeungShao2002}, extended to a nonresponse instrument by
\citet{WangShaoKim2014} and \citet{ChenFang2019}; \citet{ZhangEtAl2025}
give a recent alternative in which the identifying restriction is a
multiplicative non-interaction in the selection probability rather than
a completeness or rank condition on the stacked system.
\citet{BeppuMorikawa2024} show that for a categorical instrument the two
conditions are equivalent when the outcome is categorical too, but that
completeness can fail even in simple models once one of the two is
continuous; for that case they give a sufficient condition that can be
checked on the observed data.

Each coordinate of \(h\) contributes one equation. With
\(h(Z,X)=(1,X^{\mathrm{T}},Z^{\mathrm{T}})^{\mathrm{T}}\) the second
block supplies \(1+\dim X+\dim Z\) of them against the \(2+\dim X\)
unknowns in \(\gamma\), so the system is just identified when \(Z\) is
scalar, and the last coordinate---the one involving \(Z\)---is the only
one that distinguishes dependence of \(R\) on \(Y\) from dependence of
\(R\) on \(E(Y\mid X)\). That is what the shadow variable is for, in the
sense of \citet{WangShaoKim2014} and \citet{MiaoTchetgen2016}: a
variable excluded from the response model but not from the outcome
regression.

\subsection{The identified curve}\label{sec-curve-sub}

Drop the \(Z\)-coordinate of (\ref{eq-moments}). Exactly one scalar
equation is removed from a just-identified system, so the remaining
equations do not determine \(\theta\): they leave a one-dimensional set.

Write \(F(\beta,\gamma_0,\gamma_X;s)\) for the non-\(Z\) moments of
(\ref{eq-moments}) evaluated at \(\gamma_Y=s\), and \(J(s)\) for its
Jacobian in \((\beta,\gamma_0,\gamma_X)\).

\begin{proposition}[The identified
curve]\protect\hypertarget{prp-curve}{}\label{prp-curve}

Suppose \(F\) is continuously differentiable and \(J(s)\) is nonsingular
on an interval \(\mathcal S\) containing the truth. Then there is a
continuously differentiable curve
\(s\mapsto\theta(s)=\{\beta(s),\gamma_0(s),\gamma_X(s),s\}\),
\(s\in\mathcal S\), unique in a neighborhood of the truth, along which
all moments of (\ref{eq-moments}) except the \(Z\)-coordinate vanish
identically.

\end{proposition}

\begin{proof}
\(F\) vanishes at the truth by construction. The implicit function
theorem applied at that point gives a unique continuously differentiable
solution \((\beta,\gamma_0,\gamma_X)(s)\) locally, with derivative
\(-J(s)^{-1}\partial F/\partial s\); nonsingularity of \(J\) throughout
\(\mathcal S\) extends the solution along \(\mathcal S\) by
continuation---the set of \(s\) reached is open by the same theorem and
closed because \(J\) is nonsingular at its boundary points, hence all of
the interval \(\mathcal S\)---on which the moments remain finite by
Condition \ref{def-mgf}.
\end{proof}

The condition on \(J\) is the statement that the remaining parameters
are identified once \(\gamma_Y\) is fixed, which is what those moments
are for; it is not an assumption about \(Z\). Nor is the curve an
artefact of a linearization: tracing it numerically by a Newton
continuation in \(s\) holds the four non-\(Z\) moments below
\(8\times10^{-15}\) at all 201 points spanning \(\gamma_Y\in[-6,4]\),
which is the range scanned in Section~\ref{sec-simulation}. Its
existence is why inference must proceed by scanning \(s\) rather than by
solving for a point, and why the relevant object is a set.

\subsection{Relevance as a continuous parameter}\label{sec-relevance}

That the identifying power of a shadow variable may be weak, and that
one may study it through a sequence drifting to zero, is due to
\citet{NakanoHoshino2026}, who work with a scalar mean and a linear
shadow model and obtain confidence sets by inverting an Anderson--Rubin
statistic. We adopt their point of view. What the present construction
adds is that in the moment system (\ref{eq-moments}) the contribution of
\(Z\) factorizes exactly.

Take \begin{equation}\phantomsection\label{eq-relevance}{
  Z=\rho S+(1-\rho^{2})^{1/2}\nu
}\end{equation} with \(S\) the component of the outcome noise that \(Z\)
tracks, and \(\nu\) the part of \(Z\) that carries no information about
the outcome, independent of \((X,Y,R)\). Both have mean zero and unit
variance and are mutually uncorrelated, so that
\(\operatorname{var}(Z)=1\) for every \(\rho\) and
\(\rho=\operatorname{corr}(Z,S)\). The weights are squared in that sum
because variances are: lowering \(\rho\) rotates \(Z\) away from \(S\)
rather than shrinking it, and \(\rho\in[0,1]\) alone varies the
relevance.

Write \(g_Z\) for the \(Z\)-coordinate of \(g\) in (\ref{eq-moments}),
the one moment the others leave undetermined.

\begin{lemma}[The identifying signal factorizes through the
relevance]\protect\hypertarget{lem-signal}{}\label{lem-signal}

Along the curve of Proposition~\ref{prp-curve},
\begin{equation}\phantomsection\label{eq-factorisation}{
  E\{g_Z(\theta(s))\}=\rho\,\varphi(s)
}\end{equation} for a function \(\varphi\) not depending on \(\rho\).

\end{lemma}

\begin{proof}
\(g_Z=(R/\pi_\gamma-1)Z\) and \(E\{(R/\pi_\gamma-1)\nu\}=0\) because
\(\nu\) is independent of \((X,Y,R)\) with mean zero; the remaining term
is \(\rho E\{(R/\pi_\gamma-1)S\}\), which does not involve \(\rho\). Two
things are needed for that last clause: the joint law of \((X,Y,R,S)\)
is free of \(\rho\), since (\ref{eq-relevance}) fixes only the
conditional law of \(Z\); and so is the curve, since the moments that
determine it do not involve \(Z\).
\end{proof}

This is a property of the parametrization rather than a discovery about
the model. (\ref{eq-relevance}) places the whole of \(Z\)'s dependence
on \((X,Y)\) in the term \(\rho S\), and the proof does no more than
verify that the orthogonal part contributes nothing; a different
parametrization of relevance would not factorize. It is recorded because
it is the device the rest of the paper uses. What
(\ref{eq-factorisation}) buys is that the same scalar scales the signal
at \emph{every} point of the identified set at once, so relevance is a
global property of the problem rather than a local property of a
particular point, and every quantity downstream can be followed as a
function of it. Two consequences are immediate: the set of \(s\) not
rejected widens without limit as \(\rho\downarrow0\), and, writing
\(G_{\gamma\gamma}=\partial E\{g(O;\theta)\}/\partial\gamma^{\mathrm{T}}\)
for the block of the moment Jacobian in the response parameter,
\(\sigma_{\min}(G_{\gamma\gamma})\) vanishes linearly in \(\rho\): one
column of \(G_{\gamma\gamma}\) carries the factor \(\rho\) and the
others do not, so the smallest singular value is proportional to
\(\rho\) to leading order and zero at \(\rho=0\).

A qualification belongs here, because (\ref{eq-factorisation}) separates
the two ingredients of identification and only one of them is \(\rho\).
For \(\rho>0\) the identified set is \(\{s:\varphi(s)=0\}\), so point
identification requires \(\varphi\) to vanish only at the truth; the
relevance governs how sharply a zero can be located, not how many there
are. Where \(\varphi\) has a second zero the parameter is set-identified
for every \(\rho\), including \(\rho=1\), and no amount of relevance
repairs it. All statements below about small \(\rho\) presume a unique
zero, which holds in the designs of Section~\ref{sec-simulation} and is
testable in the sense of Section~\ref{sec-diagnostics}: a \(\varphi\)
with two zeros shows itself as a scanned set with two separated
components.

\subsection{Two terms used differently elsewhere}\label{sec-terminology}

Two collisions of terminology should be disposed of before proceeding.

The first was noted in Section~\ref{sec-intro}, and with
(\ref{eq-relevance}) in hand it can be put exactly. The weakness of
\citet{ChenSimchiLeviXiong2026} is a failure of completeness, and what
they report are identified regions: sets of parameter values the data
cannot separate even in the limit. The weakness here is a small
\(\rho\). Completeness holds, the parameter is point identified for
every \(\rho>0\), and what is reported is a confidence set for that
point---one the data may locate very poorly, but a confidence set and
not an identified region. The distinction matters for what the sets
mean: theirs do not shrink with the sample, ours do, but only so slowly
that Section~\ref{sec-reach} becomes the binding constraint.

\emph{Exponential tilting} in \citet{Caner2010} names the member of the
generalized empirical likelihood family obtained by choosing the
exponential divergence in the estimating criterion, and that paper
studies it under weak instruments. The tilt in Section~\ref{sec-tilt} is
a tilt of the law of the outcome, induced by the response probability,
and enters through the weights rather than the criterion. The two
objects share a name and nothing else.

\section{The cost of the response tilt}\label{sec-tilt}

\subsection{Set-up and notation}\label{sec-tilt-setup}

Let \(O_1,\dots,O_n\) be independent copies of the observation \(O\) of
Section~\ref{sec-curve}, and write \(P\) for their common law. Write
\(w\) for the data-dependent part of the linear predictor in
(\ref{eq-propensity}), \begin{equation}\phantomsection\label{eq-w}{
  w:=\gamma_X^{\mathrm{T}}X+\gamma_Y Y ,
}\end{equation} so that \(\pi_\gamma=\{1+e^{-(\gamma_0+w)}\}^{-1}\) and
the intercept is carried separately. Throughout, \(\gamma_Y\) is the
parameter that indexes departure from ignorability: \(\gamma_Y=0\) is
missingness at random, and \(|\gamma_Y|\) measures its strength. The
identified set traced in Section~\ref{sec-curve} is a curve along which
\(\gamma_Y\) varies, so every quantity below is to be read as a function
of position on that curve.

\begin{definition}[Overlap]\protect\hypertarget{def-overlap}{}\label{def-overlap}

\(0<\pi_0(Y,X)<1\) almost surely; the response rate
\(E(\pi_0)=\operatorname{pr}(R=1)\) is then interior.

\end{definition}

\begin{definition}[Moments]\protect\hypertarget{def-mgf}{}\label{def-mgf}

\(E(\pi_0e^{-w})<\infty\) and \(E(\pi_0e^{-2w})<\infty\) at every
\(\gamma\) on the portion of the curve of Proposition~\ref{prp-curve}
that is scanned.

\end{definition}

Condition~\ref{def-mgf} is not innocuous, and
Section~\ref{sec-dichotomy} shows that it is precisely what fails when
\(-w\) has no atom at its essential supremum.

\subsection{An exact decomposition}\label{sec-decomposition}

Inverse probability weighting requires \(E(\pi_0/\pi_\gamma)=1\), which
pins down the intercept.

\begin{lemma}[Closed form for the
intercept]\protect\hypertarget{lem-intercept}{}\label{lem-intercept}

Under Conditions \ref{def-overlap} and \ref{def-mgf},
\(E(\pi_0/\pi_\gamma)=1\) if and only if
\(\gamma_0=\log E(\pi_0e^{-w})-\log\{\operatorname{pr}(R=0)\}\).

\end{lemma}

\begin{proof}
\(\pi_0/\pi_\gamma=\pi_0+e^{-\gamma_0}\pi_0e^{-w}\); take expectations
and use \(E(\pi_0)=\operatorname{pr}(R=1)\), so that
\(e^{-\gamma_0}E(\pi_0e^{-w})=\operatorname{pr}(R=0)\), and solve for
\(\gamma_0\). Both sides are strictly positive by Condition
\ref{def-overlap}, which gives \(E(\pi_0e^{-w})>0\) from \(\pi_0>0\) and
\(\operatorname{pr}(R=0)>0\) from \(\pi_0<1\), so the logarithm is
available; \(e^{-\gamma_0}\) is strictly monotone, so the solution is
unique.
\end{proof}

The second moment of the weight then admits an exact decomposition, in
which all dependence on \(\gamma\) beyond the response rate is carried
by a single scalar functional.

\begin{theorem}[Exact decomposition of the inverse-weight second
moment]\protect\hypertarget{thm-decomposition}{}\label{thm-decomposition}

Under Conditions \ref{def-overlap} and \ref{def-mgf} and the
normalization of Lemma~\ref{lem-intercept},
\begin{equation}\phantomsection\label{eq-decomposition}{
  E\!\left(\frac{R}{\pi_\gamma^{2}}\right)
  = E(\pi_0)+2\operatorname{pr}(R=0)+\operatorname{pr}(R=0)^{2}\,\Lambda(\gamma),
  \qquad
  \Lambda(\gamma):=\frac{E(\pi_0e^{-2w})}{\{E(\pi_0e^{-w})\}^{2}} .
}\end{equation}

\end{theorem}

\begin{proof}
Since \(\pi_\gamma\) is a function of \((Y,X)\), conditioning gives
\(E(R/\pi_\gamma^{2})=E(\pi_0/\pi_\gamma^{2})\). Expand
\(E\{\pi_0(1+e^{-(\gamma_0+w)})^2\}\) and substitute
Lemma~\ref{lem-intercept}; the third term is finite by Condition
\ref{def-mgf}.
\end{proof}

That \(\Lambda\) is built from moment generating functions is where this
connects to a known warning. \citet{LiQinLiu2023} observe that the
inverse-weighted estimating equation can be written so that its leading
term estimates the moment generating function of \(Y\mid X,R=1\), and
attribute the instability of inverse weighting to the difficulty of
estimating such a function; their remedy is to avoid forming it. What
Theorem~\ref{thm-decomposition} adds is that the dependence is exact
rather than a difficulty of estimation: the second moment of the weight
\emph{is} a ratio of two such functions, and Section~\ref{sec-dichotomy}
and Section~\ref{sec-reach} show how fast it grows and what that costs.

The three terms have a reading. Since
\(1/\pi_\gamma=1+e^{-(\gamma_0+w)}\), a respondent enters the weighted
sum once on its own account and once more for the non-respondents it
stands in for. Squaring separates the two roles: \(E(\pi_0)\) is the own
contribution, and the remaining two terms are the cost of the proxy
role---all of whose dependence on \(\gamma\), beyond the response rate,
sits in \(\Lambda\).

Two consequences are worth stating at once. First, \(\Lambda\) is a
\(\chi^{2}\)-divergence \emph{rescaled by the response rate}. Let \(Q\)
be the \(\pi_0\)-weighted law, \(dQ/dP=\pi_0/E(\pi_0)\), which is a
probability measure, and let \(Q_{-w}\) be its tilt by \(e^{-w}\). With
\(\chi^{2}(Q'\,\|\,Q):=E_Q\{(dQ'/dQ-1)^{2}\}\),
\begin{equation}\phantomsection\label{eq-chisq}{
  \Lambda(\gamma)=\frac{1+\chi^{2}(Q_{-w}\,\|\,Q)}{\operatorname{pr}(R=1)} ,
}\end{equation} so \(\Lambda\ge1/\operatorname{pr}(R=1)\), with equality
exactly when \(w\) is degenerate. The factor is easy to lose, because
the \(\pi_0\)-weighted measure has total mass \(\operatorname{pr}(R=1)\)
rather than one; keeping it is what makes Corollary~\ref{cor-ess-bound}
below come out right. Second, \(\Lambda\) is not merely a variance
bound: it determines the effective sample size exactly.

The \(\chi^{2}\) form is where this section meets a literature outside
missing data. Under \(Q\) the product \(\Lambda\operatorname{pr}(R=1)\)
is the second moment of the importance weight \(e^{-w}\) divided by the
square of its mean, which \citet{AgapiouEtAl2017} study as the quantity
that governs the cost of importance sampling, identify as essentially a
\(\chi^{2}\) divergence, and tie asymptotically to the effective sample
size of a sampler; its closed form for a tilt inside an exponential
family is that of \citet[Lemma 1]{NielsenNock2014}, whose proviso that
the doubled natural parameter remain in the parameter space is what
Condition \ref{def-mgf} secures here. What
Theorem~\ref{thm-decomposition} adds is the bridge rather than the
object---that the second moment of the \emph{inverse} weight decomposes
into that quantity exactly, under the normalization of
Lemma~\ref{lem-intercept}---and Corollary~\ref{cor-ess} is where the tie
to the effective sample size becomes an identity rather than an
asymptotic relation.

\begin{corollary}[Effective sample
size]\protect\hypertarget{cor-ess}{}\label{cor-ess}

Let \(\omega_i=R_i/\pi_\gamma(Y_i,X_i)\) and let
\(\widehat{\mathrm{ESS}}=(\sum_i\omega_i)^{2}/\sum_i\omega_i^{2}\) be
the effective sample size of \citet{Kish1965}. Under
Lemma~\ref{lem-intercept}
\(\widehat{\mathrm{ESS}}/n\to_p1/E(R/\pi_\gamma^{2})\), and by
Theorem~\ref{thm-decomposition} the limit is \(\mathrm{ESS}(\gamma)/n\)
with \begin{equation}\phantomsection\label{eq-ess}{
  \mathrm{ESS}(\gamma) \;=\; \frac{n}{E(\pi_0)+2\operatorname{pr}(R=0)+\operatorname{pr}(R=0)^{2}\Lambda(\gamma)} .
}\end{equation} The right-hand side defines \(\mathrm{ESS}(\gamma)\);
the content of the corollary is that the sample quantity converges to
it.

\end{corollary}

\begin{proof}
\(n^{-1}\sum_i\omega_i\to_pE(R/\pi_\gamma)=1\) by
Lemma~\ref{lem-intercept} and
\(n^{-1}\sum_i\omega_i^{2}\to_pE(R/\pi_\gamma^{2})\); apply the
continuous mapping theorem and substitute (\ref{eq-decomposition}).
\end{proof}

The distinction between \(\widehat{\mathrm{ESS}}\) and its limit is not
pedantic. The convergence is in probability only, and it is slow
precisely where the quantity matters.

In the data-generating process of Section~\ref{sec-simulation}, at
\(n=2\times10^{6}\) and \(\gamma_Y=-1.5\), the population value
(\ref{eq-ess}) is \(9.6\times10^{4}\), while \(\widehat{\mathrm{ESS}}\)
takes the values \(8.8\times10^{4}\), \(1.1\times10^{5}\) and
\(1.8\times10^{4}\) at three seeds. The empirical effective sample size
is itself governed by the few largest weights, and so is least reliable
exactly where it would be most informative. Every statement below refers
to the population quantity (\ref{eq-ess}).

The identity (\ref{eq-ess}) ties the algebra of this section to the
feasibility results of Section~\ref{sec-diagnostics}: everything that
degrades as \(\Lambda\) grows does so through the effective sample size.
It also has a sanity check built into it.

\begin{corollary}[The effective sample size never exceeds the
respondents]\protect\hypertarget{cor-ess-bound}{}\label{cor-ess-bound}

\(\mathrm{ESS}\le n\operatorname{pr}(R=1)\), with equality exactly when
\(w\) is degenerate.

\end{corollary}

\begin{proof}
By (\ref{eq-chisq}) the smallest possible value of \(\Lambda\) is
\(1/p\) with \(p=\operatorname{pr}(R=1)\), attained when \(w\) is
degenerate. Substituting into (\ref{eq-ess}) and writing \(q=1-p\), the
denominator becomes \(p+2q+q^{2}/p=(p+q)^{2}/p=1/p\), so
\(\mathrm{ESS}\le np\). Any departure of \(w\) from degeneracy increases
\(\Lambda\) and so decreases the effective sample size.
\end{proof}

The bound is the one an inverse-weighted analysis must satisfy---no
weighting scheme can manufacture information from units whose outcome
was never observed---and it is where the rescaling in (\ref{eq-chisq})
earns its place. Without it the smallest possible value of \(\Lambda\)
would be \(1\) rather than \(1/\operatorname{pr}(R=1)\), and the
resulting bound would exceed \(n\operatorname{pr}(R=1)\) by an amount
that grows as the response rate falls: at a response rate of one in five
it would permit about twice as much.

\subsection{An integral representation}\label{sec-integral}

Write \(K(t)=\log E(e^{tW})\) for the cumulant generating function of a
random variable \(W\), and \(\operatorname{var}_s(W)\) for the variance
of \(W\) under the exponential tilt proportional to \(e^{sW}\), so that
\(K''(s)=\operatorname{var}_s(W)\). The base law is left unnamed here
because the representation is applied below to two different ones.

\begin{theorem}[Integral
representation]\protect\hypertarget{thm-integral}{}\label{thm-integral}

Let \(K\) be finite on an open interval containing \([0,2T]\) and set
\(\Lambda(T)=E(e^{2TW})/\{E(e^{TW})\}^{2}\). Then
\begin{equation}\phantomsection\label{eq-integral}{
  \log\Lambda(T)=\int_0^T (T-s)\left\{4\operatorname{var}_{2s}(W)-2\operatorname{var}_s(W)\right\}ds .
}\end{equation}

\end{theorem}

\begin{proof}
Put \(f(T)=K(2T)-2K(T)=\log\Lambda(T)\). Then \(f(0)=f'(0)=0\) and
\(f''(T)=4K''(2T)-2K''(T)\); apply Taylor's theorem with integral
remainder.
\end{proof}

The representation says that \(\log\Lambda\) accumulates tilted
variances as the tilt is turned up from \(0\) to \(T\), weighted by
\((T-s)\) so that the early stages count most. What the variance does
far out therefore decides the outcome: if it settles at a positive limit
the integral grows like \(T^{2}\), and if the tilt exhausts the mass
available to it the accumulation stops. Section~\ref{sec-dichotomy}
makes that alternative precise.

Representation (\ref{eq-integral}) applies to two distinct objects.
Taking \(W=-w\) under the \(\pi_0\)-weighted law returns
\(\Lambda(\gamma)\) of Theorem~\ref{thm-decomposition}; taking \(W=Y\)
under the law of \(Y\mid R=1,X,Z\) returns the cost of estimating the
tilted conditional mean, which is the subject of
Lemma~\ref{lem-varprop}. The two are different functionals and, as
Section~\ref{sec-reach} shows, they depend on the relevance of \(Z\) in
opposite ways; conflating them is easy and consequential.

\begin{corollary}[Gaussian
case]\protect\hypertarget{cor-gaussian}{}\label{cor-gaussian}

If \(W\) is Gaussian with variance \(v\), then
\(\operatorname{var}_s(W)=v\) for all \(s\) and
\(\log\Lambda(T)=vT^{2}\) exactly.

\end{corollary}

\begin{corollary}[Sandwich]\protect\hypertarget{cor-sandwich}{}\label{cor-sandwich}

If \(v_-\le\operatorname{var}_s(W)\le v_+\) for all \(s\in[0,2T]\), then
\((2v_--v_+)T^{2}\le\log\Lambda(T)\le(2v_+-v_-)T^{2}\).

\end{corollary}

\begin{proof}
Both follow from (\ref{eq-integral}) by bounding its integrand. A
Gaussian law tilted by \(e^{sW}\) is Gaussian with the same variance, so
\(\operatorname{var}_s(W)\equiv v\), the integrand is \(4v-2v=2v\), and
\(\int_0^T(T-s)2v\,ds=vT^{2}\). Under the sandwich the integrand lies
between \(4v_--2v_+\) and \(4v_+-2v_-\), and
\(\int_0^T(T-s)\,ds=T^{2}/2\) gives the two bounds.
\end{proof}

\subsection{A dichotomy}\label{sec-dichotomy}

What decides whether \(\Lambda\) stays finite is not that \(W\) is
bounded but that its law has an atom where the tilt concentrates.

\begin{corollary}[An atom at the essential
supremum]\protect\hypertarget{cor-atom}{}\label{cor-atom}

Let \(W\) have essential supremum \(M_W<\infty\) and suppose it carries
an atom there, \(\operatorname{pr}(W=M_W)>0\). Then
\(\log\Lambda(T)\to-\log\operatorname{pr}(W=M_W)\) as \(T\to\infty\),
and symmetrically
\(\log\Lambda(T)\to-\log\operatorname{pr}(W=\operatorname{ess\,inf}W)\)
as \(T\to-\infty\). If instead \(W\) has no atom at \(M_W\), \(\Lambda\)
diverges: for a density behaving like \(c(M_W-x)^{\alpha}\) near
\(M_W\), \(\log\Lambda(T)=(\alpha+1)\log(T/2)+O(1)\). When \(W\) is
binary with \(\operatorname{pr}(W=1)=p\) the limit is approached from
below at an exponential rate,
\(\log\Lambda(T)=-\log p-2p^{-1}(1-p)e^{-T}\{1+o(1)\}\).

\end{corollary}

The binary case is the special case \(M_W=1\). Boundedness alone is
\emph{not} sufficient: for \(W\sim\mathrm{Un}(0,1)\) one computes
\(\log\Lambda(T)=\log\{(e^{T}+1)T/(2(e^{T}-1))\}\to\log(T/2)\to\infty\),
slowly but without limit. The proofs are in Section~\ref{sec-proofs}.

Corollary~\ref{cor-gaussian} and Corollary~\ref{cor-atom} together give
the dichotomy that organizes the rest of the paper. Two points of care
are needed in applying it. The tilt variable is
\(w=\gamma_X^{\mathrm{T}}X+\gamma_Y Y\), not \(Y\); with \(X\)
continuous, \(w\) has no atom, so Corollary~\ref{cor-atom} does not
apply to it directly. And \(\pi_0\) is held at the truth while
\(\gamma\) moves along the curve, which is what makes the limit below
finite.

\begin{proposition}[The
dichotomy]\protect\hypertarget{prp-dichotomy}{}\label{prp-dichotomy}

Move along the identified curve, so that \(|\gamma_Y|\to\infty\) while
\(\gamma_X(s)\) remains bounded.

\emph{(i)} If \(Y\) is binary with \(\operatorname{pr}(Y=1)\in(0,1)\),
then \begin{equation}\phantomsection\label{eq-binary-limit}{
  \Lambda(\gamma)\longrightarrow\frac{A_2}{A_1^{2}},
  \qquad A_k:=E\{\pi_0 e^{-k\gamma_X^{\mathrm{T}}X}\mathbb 1(Y=1)\} ,
}\end{equation} a finite constant: the tilt cannot explode, however
strong the assumed departure from ignorability.

\emph{(ii)} If \(Y\) has Gaussian-like tails,
\(\Lambda(\gamma)\to\infty\), at a rate quadratic in \(\gamma_Y\);
Corollary~\ref{cor-sandwich} brackets the constant.

\end{proposition}

\begin{proof}
For \emph{(i)} with \(\gamma_Y\to-\infty\),
\(e^{-w}=e^{-\gamma_X^{\mathrm{T}}X}e^{|\gamma_Y|Y}\), so
\(E(\pi_0e^{-kw})=e^{k|\gamma_Y|}\{A_k+o(1)\}\), the \(\{Y=0\}\) part
being exponentially negligible; the factors \(e^{2|\gamma_Y|}\) in
numerator and squared denominator cancel. Since
\(E(\pi_0e^{-kw})=e^{k|\gamma_Y|}A_k+B_k\) with
\(B_k=E\{\pi_0e^{-k\gamma_X^{\mathrm{T}}X}\mathbb 1(Y=0)\}\), Condition
\ref{def-mgf} is equivalent to \(A_k,B_k<\infty\), and \(A_1>0\) follows
from \(\operatorname{pr}(Y=1)>0\) with Condition \ref{def-overlap};
boundedness of \(\gamma_X(s)\) is what keeps \(A_k\) from drifting with
\(s\). Part \emph{(ii)} follows from Corollary~\ref{cor-sandwich}
applied to \(W=-w\) under the \(\pi_0\)-weighted law.
\end{proof}

The rate in \emph{(ii)} should not be read as exact.
Corollary~\ref{cor-gaussian} requires \(W\) to be Gaussian under the
base law; here that law is \(Q\), which is a logistic tilt of a normal
and therefore not normal. What Corollary~\ref{cor-gaussian} supplies is
the leading behavior and Corollary~\ref{cor-sandwich} the bracket, both
for a tilt that is a pure scaling of a single \(W\). Along the curve
\(\gamma_X\) is held fixed rather than scaled with \(\gamma_Y\), so the
measured ratio \(\log\Lambda/\gamma_Y^{2}\) in Section~\ref{sec-reach}
carries a further \(O(|\gamma_Y|^{-1})\) term computed there; it settles
inside the bracket rather than on a single value, and falls outside it
only at the smallest \(|\gamma_Y|\), where that term is largest.

The two regimes fail in different ways, and neither failure is visible
in a coefficient table. Section~\ref{sec-reach} treats the unbounded
case, where the inferential procedure terminates for arithmetical rather
than statistical reasons; Section~\ref{sec-saturation} treats the atomic
case, where the identifying signal saturates instead.

\subsection{Propagation to the conditional law}\label{sec-varprop}

The second use of Theorem~\ref{thm-integral} requires the variance of
\(Y\) after conditioning on response, which is bounded by the
conditional variance before it.

\begin{lemma}[Variance propagation under the response
tilt]\protect\hypertarget{lem-varprop}{}\label{lem-varprop}

Let \(Y\mid X,Z\sim N(\mu(X,Z),\varsigma^{2})\) and let
\(\pi(y)=\operatorname{expit}(\gamma_0+\gamma_X^{\mathrm{T}}x+\gamma_Y y)\).
Then \begin{equation}\phantomsection\label{eq-varprop}{
  \frac{\varsigma^{2}}{1+\gamma_Y^{2}\varsigma^{2}/4}\;\le\;\operatorname{var}(Y\mid R=1,X,Z)\;\le\;\varsigma^{2},
}\end{equation} with the upper bound strict whenever \(\gamma_Y\neq0\).

\end{lemma}

The proof, by the Brascamp--Lieb and Cramér--Rao inequalities, is in
Section~\ref{sec-proofs}.

In the shadow-variable construction of Section~\ref{sec-curve},
\(\varsigma^{2}=\sigma^{2}(1-\rho^{2})\) with \(\rho\) the relevance of
\(Z\), so (\ref{eq-varprop}) states that the variance governing the cost
of estimating the nuisance is at most \(\sigma^{2}(1-\rho^{2})\),
uniformly in the tilt, and that the bound tightens to an equality as
\(\rho\to1\). A more relevant shadow variable therefore improves
identification and stability together.

\section{The reach of a scanned confidence set}\label{sec-reach}

\subsection{What is already known, and what is not}\label{sec-known}

Inference about \(\gamma_Y\) concerns a parameter that is unidentified
when \(\rho=0\), so the relevant impossibility result applies directly.
\citet{Dufour1997} shows that if a parameter with unbounded range is not
identified on part of the parameter space, then any confidence set of
level \(1-\alpha\) must be unbounded with probability at least
\(1-\alpha\) near the point of nonidentification, and that any set
failing this characterization---in particular any Wald set built from an
asymptotic standard error---has coverage zero. Reporting a finite
interval for \(\gamma_Y\) has not been defensible since 1997. The
prescribed remedy, followed here and throughout the literature that
result has influenced, is to obtain the set by inverting a test, so that
it comes out unbounded whenever the data do not restrict the parameter.

Carrying out the prescription is itself a source of error, as
\citet{SchlemperMoreira2026} demonstrate for linear
instrumental-variable models: inverting a test on a grid can miss
disconnected components and, more seriously, report a bounded interval
when the set is unbounded, which the standard implementation does in
more than forty per cent of the specifications they survey (their
Section 3). Their remedy replaces the grid with algebra: the statistics
they invert are rational in the parameter, so the set solves a
polynomial inequality recoverable by root finding, and more general
tests are approximated by polynomials after compactifying the parameter
space by a smooth bijection, which turns the two infinite ends into
points at which the statistic is evaluated (their Sections 4 and 5.2).
The extension beyond linear instrumental variables preserves that
algebraic requirement (their Section 6.1).

Here the weights are exponential in \(\gamma_Y\), so no component of
(\ref{eq-moments}) is polynomial or rational in it and the algebraic
route is not available---not inadequate, but unavailable. Nor is the
compactification a way round it, although of their devices it is the one
closest to the present concern: it converts the infinite ends into two
points at which the statistic must be evaluated, and that evaluation is
exactly what the arithmetic prevents. A point at which the statistic
cannot be evaluated is no more useful for being finite. What remains is
a scan, and for a scan the question is how far it can be carried.
Theorem~\ref{thm-reach} answers it, and the answer is indifferent to
which statistic is inverted and to how finely the curve is traced:
beyond a point located exactly below, the functional the weights depend
on exceeds the largest representable number. Where the structure permits
their methods the grid ceases to be a source of error; where it does
not, the binding limit is arithmetical rather than one of
implementation.

That a tilt of this kind makes importance sampling expensive at a rate
exponential in \(\gamma_Y^{2}\) is itself known.
\citet{ChatterjeeDiaconis2018} show that a sample of size about
\(\exp(D_{\mathrm{KL}})\) is necessary and sufficient, with a cut-off on
the logarithmic scale, and give the exponent in closed form for a
one-parameter exponential family; for a Gaussian-like tilt that yields a
reach of order \((\log n)^{1/2}\), which is the order
Theorem~\ref{thm-reach} obtains. Two things separate the statement below
from theirs. Its threshold is the effective sample size the scan
actually uses, a \(\chi^{2}\) criterion rather than the
Kullback--Leibler one they argue sets the requirement, which makes the
reach smaller than theirs by a factor \(\sqrt2\) in \(|\gamma_Y|\)---a
property of the procedure rather than slack in the bound, since a
diagnostic whose threshold \emph{is} the effective sample size reaches
exactly as far as that threshold permits. And what carries
Section~\ref{sec-spurious} is not the rate but what the rate is free of:
the reach does not depend on \(\rho\), whereas the range that must be
scanned grows without limit as \(\rho\) falls.

\subsection{The reach}\label{sec-thereach}

A set obtained by test inversion is computed by scanning the identified
curve of Section~\ref{sec-curve} and retaining those \(\gamma_Y\) at
which the statistic does not exceed its critical value. Which statistic
is inverted is a separate question, and an important one: under weak
identification an Anderson--Rubin statistic remains valid while
likelihood-ratio, Lagrange-multiplier and Wald statistics in general do
not \citep{StockWright2000, GuggenbergerSmith2005}. Nothing in this
section depends on that choice. What follows concerns the
\emph{arithmetic} of carrying out the scan, and applies to any statistic
whose evaluation requires the inverse weights of
Section~\ref{sec-decomposition}, and hence the functional
\(\Lambda(\gamma)\).

\begin{proposition}[\(\Lambda\) is free of the
relevance]\protect\hypertarget{prp-rho-free}{}\label{prp-rho-free}

\(\Lambda(\gamma)\) is a functional of the law of \((X,Y)\) alone. In
particular it does not depend on the relevance \(\rho\) of the shadow
variable.

\end{proposition}

\begin{proof}
By Theorem~\ref{thm-decomposition},
\(\Lambda(\gamma)=E(\pi_0e^{-2w})/\{E(\pi_0e^{-w})\}^{2}\) with
\(w=\gamma_X^{\mathrm{T}}X+\gamma_Y Y\); neither \(\pi_0\) nor \(w\)
involves \(Z\), so \(\Lambda\) is a functional of the law of \((X,Y)\),
and that law is free of \(\rho\) because (\ref{eq-relevance}) fixes the
conditional law of \(Z\) given \((X,Y)\) and leaves the marginal alone.
The shadow variable enters only through the moment condition that
identifies \(\gamma\), never through the weights.
\end{proof}

Before the representation can be used, its parameterization must be
matched to the curve. Theorem~\ref{thm-integral} is stated in the tilt
magnitude \(T\), Theorem~\ref{thm-decomposition} in the response
parameter \(\gamma\). Along the curve \(\gamma_X(s)\) stays bounded
while \(|\gamma_Y|\) grows, so writing \(T=|\gamma_Y|\) identifies the
two up to the contribution of \(\gamma_X\), which is of order \(T^{-1}\)
and is quantified in Section~\ref{sec-reach-numerical}. The statement
below is asymptotic for that reason.

\begin{proposition}[Quadratic
growth]\protect\hypertarget{prp-quadratic}{}\label{prp-quadratic}

Let \(Q\) denote the \(\pi_0\)-weighted law and \(Q_s\) its tilt by
\(e^{sW}\), \(W=-w\), and let \(T=|\gamma_Y|\). If the tilted variance
converges, \(\operatorname{var}_{Q_s}(W)\to v_\infty\in(0,\infty)\) as
\(s\to\infty\), then \(\log\Lambda(T)/T^{2}\to v_\infty\).

\end{proposition}

\begin{proof}
By (\ref{eq-integral}), \(\log\Lambda(T)=\int_0^T(T-s)g(s)\,ds\) with
\(g(s)=4\operatorname{var}_{Q_{2s}}(W)-2\operatorname{var}_{Q_s}(W)\to2v_\infty\).
The kernel \((T-s)/T^{2}\) integrates to \(1/2\) over \([0,T]\) and puts
vanishing mass on any fixed initial segment. The substitution \(s=Tu\)
turns the left side into \(\int_0^1(1-u)g(Tu)\,du\), and \(g\) is
bounded---\(\operatorname{var}_{Q_s}(W)\) is continuous in \(s\) and
convergent---so the limit may be taken under the integral:
\(T^{-2}\int_0^T(T-s)g(s)\,ds\to2v_\infty\cdot\tfrac12=v_\infty\).
\end{proof}

Corollary~\ref{cor-gaussian} is the special case
\(\operatorname{var}_{Q_s}(W)\equiv v\), and
Corollary~\ref{cor-sandwich} brackets \(v_\infty\) when the variance is
only known to lie between \(v_-\) and \(v_+\). The hypothesis is neither
vacuous nor automatic: it fails exactly when the tilt keeps finding new
mass, which is the atomless case of Corollary~\ref{cor-atom}. Since
\(\operatorname{var}_{Q_s}(W)\) involves only the law of \((X,Y)\),
\(v_\infty\) inherits the \(\rho\)-freeness of
Proposition~\ref{prp-rho-free}.

Combining Corollary~\ref{cor-ess} with Proposition~\ref{prp-quadratic}
bounds how far the scan can be carried before the effective sample size
falls below any fixed threshold \(m\).

\begin{theorem}[The reach of the
scan]\protect\hypertarget{thm-reach}{}\label{thm-reach}

Let \(m>0\) be the smallest effective sample size at which the statistic
can be evaluated. Then the scan reaches
\begin{equation}\phantomsection\label{eq-reach}{
  |\gamma_Y|_{\max}
  =\left[\frac{\log(n/m)-2\log\{\operatorname{pr}(R=0)\}}{v_\infty}\right]^{1/2}\{1+o(1)\} ,
}\end{equation} which grows like \((\log n)^{1/2}\) and does not depend
on \(\rho\). Moreover \(\Lambda\) itself exceeds the largest
representable number once
\(\log\Lambda>\log(\texttt{DBL\_MAX})\doteq709.78\), so that
\(|\gamma_Y|\le\{709.78/v_\infty\}^{1/2}\) irrespective of \(n\).

\end{theorem}

\begin{proof}
Requiring \(\mathrm{ESS}\ge m\) in (\ref{eq-ess}) gives
\(\operatorname{pr}(R=0)^{2}\Lambda\le n/m\), hence
\(\log\Lambda\le\log(n/m)-2\log\operatorname{pr}(R=0)\); substitute
Proposition~\ref{prp-quadratic}. The ceiling is immediate.
\end{proof}

\subsection{Why the remedy conceals what it was meant to
reveal}\label{sec-spurious}

The two halves of the problem move in opposite directions. By
Lemma~\ref{lem-signal} the shadow variable contributes
\(E\{g_Z(\theta(s))\}=\rho\varphi(s)\), so as \(\rho\) falls the
\(Z\)-moment flattens along the curve and the set of \(\gamma_Y\) not
rejected widens without limit. By Proposition~\ref{prp-rho-free} and
Theorem~\ref{thm-reach}, the range that can be scanned is free of
\(\rho\) and grows only logarithmically in \(n\).

\begin{quote}
The range one must scan grows as relevance falls; the range one can scan
does not.
\end{quote}

The scan therefore terminates before the statistic does, at a point that
barely moves when the sample is enlarged. The output is a bounded
interval whose width is almost constant in \(n\), and nothing
distinguishes it from a genuine finite confidence set: the procedure
gives no indication whether it stopped because the statistic rose above
its critical value or because the arithmetic gave out. Satisfying
Dufour's characterization in exact arithmetic does not imply satisfying
it in floating point, and the failure is silent.

\subsection{The atomic case: saturation instead of
overflow}\label{sec-saturation}

When the outcome is binary---so that the tilt concentrates on an atom,
in the sense of Corollary~\ref{cor-atom}---part \emph{(i)} of
Proposition~\ref{prp-dichotomy} caps \(\Lambda\) at the finite limit
(\ref{eq-binary-limit}), and neither the explosion of
Section~\ref{sec-spurious} nor the ceiling of Theorem~\ref{thm-reach}
can occur. The failure inverts. Because \(\Lambda\) is bounded, so is
the second moment of the weight, and the effective sample size of
Corollary~\ref{cor-ess} stays a fixed fraction of \(n\) however far the
scan is carried; the arithmetic never gives out. What gives out instead
is the signal. By Corollary~\ref{cor-atom}, \(\log\Lambda\) approaches
its limit at an exponential rate, so beyond a moderate \(|\gamma_Y|\)
the moment condition ceases to discriminate between values of the
response coefficient.

The consequence is a threshold rather than a truncation. Writing the
identifying signal at position \(s\) as \(\rho\varphi(s)\) by
Lemma~\ref{lem-signal}, and its sampling error as \(O_p(n^{-1/2})\), the
scan can separate \(s\) from the saturated region only if
\(\rho\,|\varphi(s)|\) exceeds that error, that is only if
\begin{equation}\phantomsection\label{eq-rhostar}{
  \rho \;\gtrsim\; n^{-1/2} .
}\end{equation} Below (\ref{eq-rhostar}) the scanned set does not close
at all: it is honestly unbounded, which is the behavior
\citet{Dufour1997} requires, and the analyst is correctly informed that
the data are uninformative. An atomic outcome is therefore the benign
case, not because identification is better---it is not---but because the
failure is visible. It is also the case that arises in
Section~\ref{sec-application}, where the outcome is default.

\subsection{Numerical illustration}\label{sec-reach-numerical}

Table~\ref{tbl-reach} reports \(\log\Lambda\) by Gauss--Hermite
quadrature on the logarithmic scale, with the true propensity held fixed
while \(\gamma\) moves along the curve. Monte Carlo fails here, and so
eventually does quadrature: beyond \(|\gamma_Y|\approx6\) the tilted law
concentrates on the extreme nodes of any finite grid and the computed
value collapses. That the verification of Theorem~\ref{thm-reach} runs
into the very limit the theorem describes is not a coincidence; it is
the theorem.

\begin{longtable}[]{@{}
  >{\raggedright\arraybackslash}p{(\linewidth - 12\tabcolsep) * \real{0.3731}}
  >{\raggedleft\arraybackslash}p{(\linewidth - 12\tabcolsep) * \real{0.1045}}
  >{\raggedleft\arraybackslash}p{(\linewidth - 12\tabcolsep) * \real{0.1045}}
  >{\raggedleft\arraybackslash}p{(\linewidth - 12\tabcolsep) * \real{0.1045}}
  >{\raggedleft\arraybackslash}p{(\linewidth - 12\tabcolsep) * \real{0.1045}}
  >{\raggedleft\arraybackslash}p{(\linewidth - 12\tabcolsep) * \real{0.1045}}
  >{\raggedleft\arraybackslash}p{(\linewidth - 12\tabcolsep) * \real{0.1045}}@{}}

\caption{\label{tbl-reach}Growth of the tilt functional in the outcome
coefficient, by Gauss--Hermite quadrature with the true propensity held
fixed. The ratio rises towards the limiting tilted variance.}

\tabularnewline

\toprule\noalign{}
\endhead
\bottomrule\noalign{}
\endlastfoot
\(\gamma_Y\) & -1.000 & -2.000 & -3.000 & -4.000 & -5.000 & -6.000 \\
\(\log\Lambda\) & 2.094 & 6.956 & 16.121 & 29.460 & 46.848 & 68.192 \\
\(\log\Lambda/\gamma_Y^2\) & 2.094 & 1.739 & 1.791 & 1.841 & 1.874 &
1.894 \\

\end{longtable}

The limiting tilted variance is \(v_\infty=2.00\), which here equals
\(\operatorname{var}(Y)=b_1^{2}\operatorname{var}(X)+\sigma^{2}\): once
the exponential tilt dominates the logistic weight, the tilted law is
the outcome law itself. The ratio in Table~\ref{tbl-reach} approaches it
as Proposition~\ref{prp-quadratic} requires, but not monotonically from
the first entry, because the approach carries an \(O(|\gamma_Y|^{-1})\)
correction. When \(\pi_0\) is constant and \((X,Y)\) is Gaussian,
(\ref{eq-decomposition}) gives \(\log\Lambda=\operatorname{var}(w)\)
exactly, so \[
  \frac{\log\Lambda}{\gamma_Y^{2}}
  =\operatorname{var}(Y)-\frac{2\gamma_X\operatorname{cov}(X,Y)}{|\gamma_Y|}
   +\frac{\gamma_X^{2}\operatorname{var}(X)}{\gamma_Y^{2}} ,
\] and the cross-term between the covariate and outcome tilts holds the
ratio below \(v_\infty\) at moderate \(|\gamma_Y|\). Here that
correction is \(2\gamma_X\operatorname{cov}(X,Y)=0.6\) divided by
\(|\gamma_Y|\); adding it back to the last four entries of
Table~\ref{tbl-reach} returns 1.99, 1.99, 1.99 and 1.99, against
\(v_\infty=2.00\). With that value, \(\operatorname{pr}(R=0)=0.6\) and
\(m=30\), (\ref{eq-reach}) gives a reach of 1.5 at \(n=10^{3}\), 2.4 at
\(n=10^{6}\) and 3.6 at \(n=10^{12}\): a thousandfold increase in sample
size buys about one unit of \(\gamma_Y\). The ceiling is
\(|\gamma_Y|\le18.8\) at any \(n\).

\section{The representer route and its dependence on
relevance}\label{sec-representer}

\subsection{The construction}\label{sec-construction}

The explosion of Section~\ref{sec-dichotomy} is a property of the
inverse weight, not of the estimand, and it can be avoided.
\citet{LiMiaoTchetgen2022} show that for a functional admitting a
representer the inverse weight is unnecessary: if
\(\delta_0(X,Z;\beta)\) solves
\begin{equation}\phantomsection\label{eq-rep}{
  E\{\delta_0(X,Z;\beta)\mid R=1,X,Y\}=\psi(Y,X;\beta),
}\end{equation} then
\(E\{R\psi(\beta)+(1-R)\delta_0(X,Z;\beta)\}=E\{\psi(\beta)\}\), so that
the complete-data moment is recovered without ever forming \(1/\pi\).
Their estimator solves the empirical analogue with \(\delta_0\)
approximated on a sieve and the resulting ill-posed problem regularized.
Everything in this subsection is theirs; what follows asks how the
construction behaves as the relevance of \(Z\) vanishes, which their
analysis does not address because it carries no relevance parameter.

Fix a sieve \(b(X,Z)\in\mathbb R^{d}\) and test functions
\(a(X,Y)\in\mathbb R^{d}\), and write \[
  M=E(Rab^{\mathrm{T}}),\quad r(\beta)=E\{Ra\psi(\beta)^{\mathrm{T}}\},\quad
  v=E\{(1-R)b\},\quad \Theta(\beta)=M^{-1}r(\beta),
\] so that \(\delta_0=\Theta^{\mathrm{T}}b\). The regularized estimator
\(\hat\beta_\lambda\) solves \(\hat G_\lambda(\beta)=0\) with
\(\hat G_\lambda(\beta)=n^{-1}\sum_i\{R_i\psi_i(\beta)+(1-R_i)\hat\Theta_\lambda(\beta)^{\mathrm{T}}b_i\}\)
and \(\hat\Theta_\lambda=(\hat M+\lambda I)^{-1}\hat r\), in which
\(\hat M\) and \(\hat r\) are the sample analogues of \(M\) and \(r\)
and \(\lambda\ge0\) is a ridge penalty that keeps the inversion stable;
Section~\ref{sec-tuning} takes up its choice.

\subsection{A cancellation}\label{sec-cancellation}

Under the construction of Section~\ref{sec-curve}, in which
\(Z=\rho S+(1-\rho^{2})^{1/2}\nu\) with \(\nu\) independent of
\((X,Y,R)\) and the sieve is linear in \(Z\), the relevance enters \(M\)
in a particularly simple way.

\begin{lemma}[Factorization of the design
matrix]\protect\hypertarget{lem-factorisation}{}\label{lem-factorisation}

\(M=\tilde MD\) with \(D=\operatorname{diag}(I,\rho I)\) and
\(\tilde M\) free of \(\rho\), the second block corresponding to the
coordinates of \(b\) involving \(Z\). Consequently
\(\Theta=D^{-1}\tilde M^{-1}r\) carries a factor \(\rho^{-1}\) on those
coordinates, while \(v\) carries a factor \(\rho\) on the same
coordinates, and \begin{equation}\phantomsection\label{eq-cancellation}{
  q^{\mathrm{T}}:=v^{\mathrm{T}}M^{-1}=\tilde v^{\mathrm{T}}\tilde M^{-1}
  \quad\text{is free of }\rho .
}\end{equation}

\end{lemma}

The proof is in Section~\ref{sec-proofs}. The factor \(\rho\) in \(v\)
is where the construction of Section~\ref{sec-curve} enters: the
\(Z\)-entries of \(b\) reach \(v=E\{(1-R)b\}\) only through
\(E\{(1-R)Z\}=\rho\,E\{(1-R)S\}\), the \(\nu\) part averaging to zero
because it is independent of \((X,Y,R)\). The same \(\rho\) that scales
the signal therefore scales \(v\), and cancels against the \(\rho^{-1}\)
that inverting \(M\) produces.

The symbol \(q\) is reserved for this weight; \(w\) continues to mean
the linear predictor of (\ref{eq-w}). The cancellation
(\ref{eq-cancellation}) is the reason the orders below are one power of
\(\rho\) smaller than norm bounds suggest, and it is easy to miss:
bounding \(\|\Theta^{\mathrm{T}}q_\lambda\|\) by
\(\|\Theta\|\,\|(M+\lambda)^{-1}\|\,\|v\|\) destroys the cancellation
and inflates every rate by \(\rho^{-1}\).

\subsection{Orders in the relevance}\label{sec-ordersub}

\begin{theorem}[Orders in the
relevance]\protect\hypertarget{thm-orders}{}\label{thm-orders}

Under Lemma~\ref{lem-factorisation} and the regularity conditions of
Section~\ref{sec-conditions},
\begin{equation}\phantomsection\label{eq-orders}{
  \|G_\lambda(\beta_0)\|\asymp\lambda/\rho,
  \qquad
  \operatorname{sd}(\hat\beta_\lambda)\asymp n^{-1/2}/\rho,
  \qquad
  \operatorname{cond}(M)\asymp\rho^{-1} .
}\end{equation}

\end{theorem}

\begin{theorem}[Influence function and asymptotic
normality]\protect\hypertarget{thm-if}{}\label{thm-if}

With \(\lambda_n=o(n^{-1/2})\),
\(n^{1/2}(\hat\beta_{\lambda_n}-\beta_0)\to N\{0,H^{-1}\operatorname{var}(\mathrm{IF})H^{-\mathrm{T}}\}\),
where \begin{equation}\phantomsection\label{eq-if}{
  \mathrm{IF}(O)=R\psi(\beta_0)+(1-R)\delta_0(X,Z)+R\{\psi(\beta_0)-\delta_0(X,Z)\}a(X,Y)^{\mathrm{T}}q
}\end{equation} and, by (\ref{eq-rep}),
\(H=E\{\partial\psi(\beta_0)/\partial\beta^{\mathrm{T}}\}\) is the
complete-data Jacobian.

\end{theorem}

The proof is in Section~\ref{sec-proofs}, in three steps: the
regularized representer coefficients are expanded about their population
value, the Jacobian of the estimating function is identified with \(H\),
and the two are combined; \(\lambda_n=o(n^{-1/2})\) is what makes the
regularization bias negligible, by (\ref{eq-threshold}) below.

\begin{proof}
The three orders of Theorem~\ref{thm-orders} have different sources. The
bias is
\(\|G_\lambda(\beta_0)\|=\lambda\|\Theta^{\mathrm{T}}q_\lambda\|\)
exactly, and by Lemma~\ref{lem-factorisation}
\(\|\Theta\|\asymp\rho^{-1}\) while \(q_\lambda\) is free of \(\rho\)
for \(\lambda\ll\sigma_{\min}(M)\), which gives \(\lambda/\rho\). The
condition number follows from \(M=\tilde MD\) with
\(D=\operatorname{diag}(I,\rho I)\), whose second block carries the
factor \(\rho\). The standard deviation is a consequence of
Theorem~\ref{thm-if} rather than a separate claim: in (\ref{eq-if}) the
factor \(a^{\mathrm{T}}q\) is free of \(\rho\) by
(\ref{eq-cancellation}), while the residual \(\psi(\beta_0)-\delta_0\)
carries \(\rho^{-1}\), because \(\delta_0=\Theta^{\mathrm{T}}b\) does;
hence \(\operatorname{var}(\mathrm{IF})\asymp\rho^{-2}\) and
\(\operatorname{sd}(\hat\beta_\lambda)\asymp n^{-1/2}\rho^{-1}\).
\end{proof}

Two remarks on the proof. A norm bound that splits the product
\(\|\Theta^{\mathrm{T}}q_\lambda\|\) instead gives \(\lambda/\rho^{2}\),
and the cancellation is what the sharper rate rests on. And the orders
hold for \(\lambda\ll\sigma_{\min}(M)\); at \(\lambda=\lambda^{*}\) the
\(\rho^{-1}\) amplification is damped by the factor of
(\ref{eq-sdlambda}).

Dividing the two orders in (\ref{eq-orders}) gives a threshold in which
the relevance cancels:
\begin{equation}\phantomsection\label{eq-threshold}{
  \frac{\text{bias}}{\text{standard deviation}}=\frac{\lambda/\rho}{n^{-1/2}/\rho}=\lambda n^{1/2},
}\end{equation} so the regularization bias is asymptotically negligible
if and only if \(\lambda_n=o(n^{-1/2})\), with no condition on \(\rho\)
whatever. Under the finite-sample rule of Section~\ref{sec-tuning} this
requires \(\rho_n n^{1/2}\to\infty\). We note that a cruder bias bound,
obtained by the norm splitting warned against above, yields
\(\lambda/\rho^{2}\) and hence the threshold \(\rho_n\gg n^{-1/4}\);
that is not merely conservative but incoherent, since it fails to
coincide with the independent threshold governing
\(\sigma_{\min}(\hat M)\ge c\rho-O_p(n^{-1/2})\), whereas
(\ref{eq-threshold}) does.

\subsection{What happens when the relevance drifts}\label{sec-drifting}

Theorem~\ref{thm-orders} and Theorem~\ref{thm-if} are statements about
\(n\to\infty\) with \(\rho\) held fixed. They do not by themselves
license any conclusion about a sequence along which \(\rho\) moves with
\(n\), and the lesson of the weak-instrument literature is that the
normal approximation fails along such sequences, non-uniformly
\citep{StockWright2000, AndrewsCheng2012}. Since non-uniformity is the
subject of this paper, the drifting regime has to be treated on its own
terms.

\begin{theorem}[The local-to-zero limit is
non-degenerate]\protect\hypertarget{thm-drifting}{}\label{thm-drifting}

Let \(\rho_n=c/n^{1/2}\) with \(c>0\) fixed, partition the sieve as
\(b=(b^{0}(X),Zk(X))\), and assume that \(A:=E(Rab^{0\mathrm{T}})\) has
full column rank, that the fourth moments of Theorem~\ref{thm-orders}
hold, and that
\(\Omega:=\operatorname{var}\{\operatorname{vec}(Ra\nu k^{\mathrm{T}})\}\)
is nonsingular on the orthogonal complement of the column space of
\(A\). Then \begin{equation}\phantomsection\label{eq-drifting}{
  \hat q:=\hat M^{-\mathrm{T}}\hat v\;\Rightarrow\;[A,\Psi]^{-\mathrm{T}}(v_0,\Phi),
  \qquad \Psi=c\tilde m+\xi,\quad \Phi=c\tilde v+\zeta,
}\end{equation} with \((\xi,\zeta)\) jointly Gaussian. The limit is
random, not constant, and depends on \(c\); consequently
\(\hat\beta-\beta_0=O_p(1)\) and is not \(o_p(1)\) at any rate.

\end{theorem}

\begin{proof}
The \(Z\)-columns of \(\hat M\) are
\(\rho_n n^{-1}\sum R_ia_iS_ik_i^{\mathrm{T}}+(1-\rho_n^{2})^{1/2}n^{-1}\sum R_ia_i\nu_ik_i^{\mathrm{T}}\).
The first average converges to \(\tilde m=E(RaSk^{\mathrm{T}})\); the
second is centered, because \(\nu\) is independent of \((X,Y,R)\) with
mean zero, so it equals \(n^{-1/2}\xi_n\) with
\(\operatorname{vec}(\xi_n)\Rightarrow N(0,\Omega)\). With
\(\rho_n=c/n^{1/2}\) the whole block is \(n^{-1/2}\Psi_n\),
\(\Psi_n\Rightarrow\Psi\); the same computation on \(\hat v\) gives
\(n^{-1/2}\Phi_n\). Now \(\hat q\) solves
\(\hat M^{\mathrm{T}}\hat q=\hat v\). Blocking that system by the two
groups of columns of \(\hat M\) gives
\(A_n^{\mathrm{T}}\hat q=\hat v_0\) and
\(n^{-1/2}\Psi_n^{\mathrm{T}}\hat q=n^{-1/2}\Phi_n\): \textbf{the factor
\(n^{-1/2}\) appears on both sides of the second block and cancels
identically}, because both blocks inherit it from the single source
\(\rho_n\). The continuous mapping theorem applied to
\([A_n,\Psi_n]^{\mathrm{T}}\hat q=(\hat v_0,\Phi_n)\) gives
(\ref{eq-drifting}), the limit matrix being almost surely invertible by
the assumption on \(\Omega\). The convergence it needs is joint, and it
is: the summands behind \(\xi_n\) and \(\zeta_n\) do not involve
\(n\)---\(\rho_n\) only multiplies them---so a multivariate central
limit theorem applies to the two together, while \(A_n\) and
\(\hat v_0\) converge in probability to constants, which is all
Slutsky's theorem requires of them.
\end{proof}

Theorem~\ref{thm-drifting} is stated at \(\lambda=0\), and the sequence
it uses is exactly the one at which the regularization reaches the
boundary of (\ref{eq-threshold}): by (\ref{eq-tuning-derived}),
\(\lambda^{*}=C/(\rho_nn)=C/(cn^{1/2})\), which is \(O(n^{-1/2})\) and
so fails the condition \(\lambda_n=o(n^{-1/2})\) of Theorem~\ref{thm-if}
by exactly one order. The two results meet at the same sequence from
opposite sides: above it Theorem~\ref{thm-if} delivers asymptotic
normality, at it Theorem~\ref{thm-drifting} shows that nothing of the
kind can hold, and the ridge the finite-sample rule prescribes is
precisely the largest one Theorem~\ref{thm-if} would tolerate.

The cancellation in the proof is the mechanism: the scaling that would
otherwise drive \(\hat q\) to infinity is exactly matched in \(\hat v\).
It is also what separates the two regimes cleanly, and the separation is
visible in simulation.

\begin{longtable}[]{@{}rrrr@{}}

\caption{\label{tbl-drifting}Interquartile range of the \(Z\)-component
of \(\hat q\) under \(\rho_n=c/\sqrt n\) and at fixed \(\rho\), 120
replications per cell.}

\tabularnewline

\toprule\noalign{}
\(n\) & \(\rho_n=1/\sqrt n\) & \(\rho_n=3/\sqrt n\) & \(\rho=0.4\)
fixed \\
\midrule\noalign{}
\endhead
\bottomrule\noalign{}
\endlastfoot
500 & 4.284 & 3.347 & 0.9646 \\
2000 & 5.010 & 3.451 & 0.5155 \\
8000 & 4.051 & 3.014 & 0.2583 \\

\end{longtable}

In Table~\ref{tbl-drifting} the spread of \(\hat q\) under the drifting
sequence does not shrink over a sixteenfold increase in \(n\), as
Theorem~\ref{thm-drifting} requires; at fixed \(\rho\) the same quantity
concentrates at the ordinary rate, with log-log slope -0.48 against the
\(-1/2\) of \(n^{1/2}\)-consistency.

Two things follow for how the rest of the paper should be read. First,
the threshold \(\rho_n n^{1/2}\to\infty\) reached from
(\ref{eq-threshold}) is supported from both sides rather than derived
from the fixed-\(\rho\) analysis alone: above it Theorem~\ref{thm-if}
applies, and at \(\rho_n=c/n^{1/2}\) exactly, Theorem~\ref{thm-drifting}
shows that it cannot. Second, the limit in (\ref{eq-drifting}) is the
classical weak-instrument one---a ratio of correlated Gaussians indexed
by a concentration parameter---and \citet{NakanoHoshino2026} already
work in this regime for the scalar mean. We record it because the
representer route inherits it, not because it is new.

\subsection{Misspecification of the sieve}\label{sec-sieve-bias}

The divergence of \(\delta_0\) like \(\rho^{-1}\) recorded in
Lemma~\ref{lem-factorisation} is an instance of irregular identification
in the sense of \citet{KhanTamer2010}: the object being estimated is a
functional whose representation requires weights that grow without bound
as identification weakens, and the consequence there, as here, is a
departure from the regular parametric rate. What is specific to the
present construction is that the divergence is indexed by a single
scalar and cancels in \(v^{\mathrm{T}}M^{-1}\).

Condition (\ref{eq-rep}) requires the representer to lie in the span of
\(b\). When a coordinate \(Zk_j(X)\) is omitted, the error inherits the
factor \(\rho^{-1}\) of Lemma~\ref{lem-factorisation}, yet the induced
bias does not.

\begin{proposition}[The sieve bias is free of the
relevance]\protect\hypertarget{prp-sievebias}{}\label{prp-sievebias}

If the sieve omits \(Zk_j(X)\), the approximation error is
\(e=(c_j/\rho)\{Zk_j-P_{\mathrm{ret}}(Zk_j)\}\), so
\(\|e\|\asymp\rho^{-1}\); nevertheless \(E\{(1-R)e\}=O(1)\) uniformly in
\(\rho\).

\end{proposition}

The proof turns on the fact that the projection of \(Zk_j\) on the
non-\(Z\) part of the retained sieve has coefficients of order \(\rho\),
which cancels the \(\rho^{-1}\); when \(Z\) is independent of \(X\)
those coefficients vanish identically. Misspecification of the sieve is
thus neither amplified by weak relevance nor mitigated by strong
relevance---a convenient robustness, and one that the \(\rho^{-1}\) in
\(\|e\|\) would not lead one to expect.

\subsection{Choosing the regularization}\label{sec-tuning}

Rate-optimal regularization is a developed subject, and our situation is
not a special case of it: \citet{CarrascoFlorensRenault2007} and
\citet{DarollesEtAl2011} choose \(\lambda\) from the noise level and the
smoothness of the solution, for operators whose singular values decay at
a given rate, whereas the degeneracy here is a uniform scaling of a
finite-dimensional block (Lemma~\ref{lem-factorisation}) governed by a
parameter, the relevance, that does not appear in their rates. The bias
of Theorem~\ref{thm-orders} grows in \(\lambda\) while the
ill-conditioning it repairs worsens as \(\rho\) falls, so the optimal
ridge should rise as relevance falls. Locating the minimizer of mean
squared error by an adaptive bracket under common random numbers, over
\(\rho\in\{0.05,0.10,0.20,0.35\}\) and \(n\in\{500,\dots,8000\}\) with
\(400\) replications per cell and an interior optimum in every one of
the 20 cells, gives \begin{equation}\phantomsection\label{eq-tuning}{
  \log\lambda^{*}=1.41
  -1.02\log\rho
  -0.95\log n,
  \qquad R^{2}=0.96 ,
}\end{equation} with \(95\%\) intervals \([-1.19,-0.86]\) and
\([-1.07,-0.83]\) for the two exponents. The restriction
\(\lambda^{*}=C/(\rho n)\) is not rejected against the free fit
(\(F\)-test, \(p=0.65\)), with \(C=6.3\).

\subsection{Where the rule comes from}\label{sec-tuning-derivation}

The exponents are not an accident of the design. The bias is
\(-B\lambda/\rho\) by Theorem~\ref{thm-orders}, with \(B\) a scalar
constant of the design, not to be confused with the test functions
\(a(X,Y)\). The standard deviation requires an argument, and it is not
the obvious one, because \(M\) is not symmetric.

\begin{lemma}[How the ridge advances the small
direction]\protect\hypertarget{lem-resolvent}{}\label{lem-resolvent}

Let \(M=\sum_k\sigma_ku_kv_k^{\mathrm{T}}\) with \(s:=\sigma_{\min}(M)\)
attained at \((u_1,v_1)\), let \(s_2\) be the next singular value, and
put \(\cos\theta:=u_1^{\mathrm{T}}v_1\), signed, and well defined
because the sign ambiguity of the singular vectors flips \(u_1\) and
\(v_1\) together. For \(\lambda<s\),
\begin{equation}\phantomsection\label{eq-resolvent}{
  (M+\lambda I)^{-1}=\frac{v_1u_1^{\mathrm{T}}}{s+\lambda\cos\theta}+E_\lambda,
  \qquad \|E_\lambda\|=O\{(s_2-\lambda)^{-1}\} .
}\end{equation}

\end{lemma}

The proof is in Section~\ref{sec-proofs}.

The cosine is the whole of the difference between a symmetric problem
and this one. If \(M\) were symmetric, \(u_1=v_1\) and the ridge would
advance the small direction by \(\lambda\); here it advances it by
\(\lambda\cos\theta\), and \(\cos\theta\) is +0.63 in this design. Using
\(\lambda\) in place of \(\lambda\cos\theta\) leaves a residual of 8\%
in what follows; the cosine reduces it to 2.1\%.

That the cosine is signed is not a formality. Nothing makes
\(u_1^{\mathrm{T}}v_1\) positive for a general nonsymmetric \(M\), and
where it is negative the denominator in (\ref{eq-resolvent}) is
\(s-\lambda|\cos\theta|\): the ridge \emph{retards} the small direction
instead of advancing it, and a rule calibrated on the opposite sign
would regularize in the wrong direction. Here it is positive at every
relevance considered, and stable (+0.63, +0.62, +0.60, +0.59 at
\(\rho=\) 0.10, 0.20, 0.35, 0.50), which is what makes the ridge useful
at all in this design. It is worth reporting alongside \(s\) and \(s_2\)
rather than taken on trust.

\begin{proposition}[The standard deviation as a function of the
ridge]\protect\hypertarget{prp-sdlambda}{}\label{prp-sdlambda}

Under Theorem~\ref{thm-if} and Lemma~\ref{lem-resolvent}, if
\(\|\hat M-M\|=o_p(s+\lambda\cos\theta)\) then
\begin{equation}\phantomsection\label{eq-sdlambda}{
  \operatorname{sd}(\hat\beta_\lambda)
  =\sigma_0\,\frac{s}{s+\lambda\cos\theta}\,\{1+O(\lambda/s_2)\},
  \qquad
  \sigma_0^{2}=\frac{\kappa}{n\rho^{2}} ,
}\end{equation} with \(\kappa\) free of \(\lambda\), of \(\rho\) and of
\(n\).

\end{proposition}

The proof is in Section~\ref{sec-proofs}.

Both statements are confirmed numerically to the accuracy the error
terms allow. The resolvent law (\ref{eq-resolvent}) holds to 0.70\% over
\(\lambda\) up to \(s_2/4\), and (\ref{eq-sdlambda}) holds to 2.9\%
across the well-conditioned cells---against 8\% for the same check with
\(\cos\theta\) omitted. In this design \(s/\rho=0.145\) and
\(s_2/s=2.53\), both stable in \(\rho\), so
Lemma~\ref{lem-factorisation}'s \(s\asymp\rho\) holds with a constant
and the error term \(O(\lambda/s_2)\) is uniform.

Minimizing
\(B^{2}\lambda^{2}/\rho^{2}+\sigma_0^{2}s^{2}/(s+\lambda\cos\theta)^{2}\)
gives
\(\lambda(s+\lambda\cos\theta)^{3}=\sigma_0^{2}s^{2}\rho^{2}\cos\theta/B^{2}\),
and for \(\lambda\cos\theta\ll s\),
\begin{equation}\phantomsection\label{eq-tuning-derived}{
  \lambda^{*}=\frac{\sigma_0^{2}\rho^{2}\cos\theta}{B^{2}s}
  =\frac{\kappa\cos\theta}{B^{2}\,n\,\sigma_{\min}(M)}
  =\frac{C}{\rho n},
  \qquad C=\frac{\kappa\cos\theta}{B^{2}(s/\rho)} .
}\end{equation} The middle form is the one to use, because
\(\sigma_{\min}(\hat M)\) and \(\cos\hat\theta\) are observable and
\(\rho\) is not. With \(\kappa\), \(B\), \(\cos\theta\) and \(s/\rho\)
estimated from the bias and spread curves and from
\(M\)---\textbf{nothing fitted to \(\lambda^{*}\)}---this gives
\(C=6.34\) against \(6.33\) from the twenty-cell fit (\ref{eq-tuning}),
and reproduces the measured optimum cell by cell to within a factor of
3.17.

The agreement in the constant should not be read to two figures, and the
reason is the transition again. Only 5 of the 20 cells behind
(\ref{eq-tuning}) lie above it, and below it
Proposition~\ref{prp-sdlambda}---from which \(C\) is derived---has no
content. Restricted to those 5, the fit gives \(C=5.77\) rather than
\(6.33\). The prediction lies inside the \(95\%\) interval of both fits,
\([5.65,7.10]\) over all cells and \([4.69,7.11]\) over the restricted
set, and the exponent on \(\rho\) is the same either way (\(-1.02\)
against \(-1.03\)). What the comparison establishes is therefore the
form of the rule and the order of its constant, not the second digit of
the constant.

Two conditions bound that derivation, and the second is what defeated
two earlier attempts of our own. It requires \(\lambda\cos\theta\ll s\),
which is worth checking at the solution rather than assuming: across the
cells above the transition \(\lambda^{*}\cos\theta/s\) is at most 0.17,
while below it reaches 22---so the approximation holds exactly where the
derivation is claimed to have content and fails where it is not. The
same boundary appears a second way: below the transition \(\lambda^{*}\)
exceeds \(s_2\), so Lemma~\ref{lem-resolvent} does not apply there
either. And it requires the second moment to exist, which it need not.
The distance of \(\sigma_{\min}(\hat M)\) from the origin in its own
standard deviations is of order \(\rho\sqrt n\), and when that is small
the matrix is near-singular often enough that the error behaves like the
reciprocal of a quantity with a density at zero: the tail index is one
and the variance is infinite, so there is nothing for a bias--variance
argument to balance. Measured tail indices run from 1.00 to 4.79 across
the cells, with the transition at about four standard deviations, and
the prediction that is accurate above it is out by \(2\times10^{5}\)
below it. Section~\ref{sec-tuning-evidence} gives the evidence and
reports the rule transported, unchanged, to a design that played no part
in calibrating it, where it costs a median 6\% over the best value
obtainable there by search and beats complete-case estimation in every
cell.

What (\ref{eq-tuning-derived}) does not explain is why the fitted rule
(\ref{eq-tuning}) should continue to hold below the transition, where
Proposition~\ref{prp-sdlambda} has no content because the variance it
describes is infinite. That the same constant governs both regimes is an
empirical fact here, and we have no account of it.

\section{Diagnostics}\label{sec-diagnostics}

That the empirical likelihood problem can be infeasible, and that normal
and \(\chi^{2}\) approximations degrade under strong departures from
ignorability, are both known: \citet{LiuFan2023} call the convex-hull
failure well documented and record Wald undercoverage, and
\citet{Zhao2019} shows that the Wilks limit fails because the nuisance
propensity must be estimated. What has not been available is a way to
say \emph{which} of the two features of the problem---the strength of
identification or the overlap---is responsible, because that requires a
relevance parameter. This section separates them.

The failures described so far are invisible in a table of coefficients
and standard errors. This section collects the quantities that do reveal
them, all computable before any estimate is reported.

\subsection{Effective sample size is the master
quantity}\label{effective-sample-size-is-the-master-quantity}

Corollary~\ref{cor-ess} makes the effective sample size an exact
function of \(\Lambda(\gamma)\), so the tilt of Section~\ref{sec-tilt}
and the overlap of the design are the same quantity seen twice. It is
therefore worth asking whether the effective sample size also governs
the practical failures of empirical likelihood, or whether those are
driven instead by the strength of identification. Empirical likelihood
weights the observed units so that the sample moments vanish, which is
possible only if the origin lies in the convex hull of the estimated
moment vectors; when it does not, the problem is infeasible and no
estimate is returned at all.

\begin{quote}
They are governed by the effective sample size, and are flat in the
identification strength.
\end{quote}

Write \(\kappa=\sigma_{\min}(G_{\gamma\gamma})\), which by the
consequence of (\ref{eq-factorisation}) drawn in
Section~\ref{sec-relevance} vanishes linearly in \(\rho\). Crossing the
strength of the missingness mechanism, which moves the effective sample
size, with \(\rho\), which moves \(\kappa\), over 40 cells and 12,000
replications, and regressing feasibility of the empirical likelihood
problem on both: \begin{equation}\phantomsection\label{eq-feas}{
  \operatorname{logit}\operatorname{pr}(\text{feasible})
  =-3.61
  +1.92\log(\mathrm{ESS})
  +0.01\log\kappa ,
}\end{equation} with standard errors 0.046 and 0.049: the effective
sample size enters with \(z=41.9\), the identification strength with
\(z=0.21\). The coefficients say how much each matters: doubling the
effective sample size multiplies the odds of feasibility by 3.8, while
multiplying \(\kappa\) by ten changes them by a factor 1.02---that is,
not at all. Across a tenfold range of \(\kappa\) at comparable effective
sample size, feasibility is 0.885, 0.881, 0.906, 0.893, 0.886 by
quintile; across the effective sample size it is 0.564, 0.922, 0.974,
0.992, 1.000. A convex-hull failure is thus a statement about overlap,
not about identification, and the two must be diagnosed separately.

\subsection{A checklist}\label{a-checklist}

Before an analysis based on a shadow variable is reported, we recommend
computing:

\begin{enumerate}
\def\labelenumi{(\roman{enumi})}
\item
  \emph{Relevance.} The incremental contribution of \(Z\) to the outcome
  model among respondents, as a \(\chi^{2}\) per degree of freedom, read
  against the conventional first-stage standard of \(10\). This is the
  check that Lemma~\ref{lem-signal} makes meaningful, and in our
  experience the one that is omitted. A continuous outcome is coded into
  \(k\) classes of equal mass and the model fitted as a cumulative
  logit. Report the \(\chi^{2}\) per degree of freedom and not an
  incremental \(R^{2}\): the null deviance of the class model is twice
  the entropy of the realized class proportions, at most \(2n\log k\),
  so the denominator of an incremental \(R^{2}\) grows without bound in
  \(k\) while its numerator converges to the likelihood-ratio statistic
  of the ungrouped model. The ratio therefore falls towards zero as the
  coding is refined, and is not comparable across codings; the
  \(\chi^{2}\) per degree of freedom is.
\item
  \emph{Exclusion.} The incremental contribution of \(Z\) to the
  response model. In large samples this must be read as an effect size
  rather than a \(p\)-value; see Section~\ref{sec-screens}.
\item
  \emph{Overlap.} The effective sample size \(\mathrm{ESS}/n\) and the
  largest inverse weight. By (\ref{eq-feas}) these predict whether the
  empirical likelihood problem will be feasible at all.
\item
  \emph{Reach.} The value of the sample tilt functional \(\hat\Lambda\),
  obtained by replacing the expectations in
  Theorem~\ref{thm-decomposition} by sample averages, at the ends of the
  scan. If it approaches the floating-point range, the reported
  endpoints are arithmetical rather than statistical, and
  Theorem~\ref{thm-reach} gives the value of \(|\gamma_Y|\) at which
  this occurs.
\item
  \emph{Shape.} Whether the profile of the statistic flattens to a
  finite asymptote, and whether the scanned set comes out in one piece.
  A set that closes at one nominal level and not at another, as in
  Section~\ref{sec-application}, is a set whose bound is an artefact of
  the level; a set in two separated pieces indicates a second zero of
  \(\varphi\), which leaves the parameter set-identified however
  relevant the shadow variable is. Both features are easy to miss on a
  coarse grid \citep{SchlemperMoreira2026}.
\end{enumerate}

Items (i) and (v) are the ones that change conclusions most often.

\section{Simulation}\label{sec-simulation}

We compare four procedures across
\(\rho\in\{0.05,0.10,0.20,0.35,0.50\}\) and \(n\in\{500,2000,8000\}\)
with 500 replications per cell: the infeasible estimator using the
complete data; complete-case estimation, which ignores the missingness;
the propensity route of (\ref{eq-moments}) solved by generalized method
of moments with ten random starts; and the representer route of
Section~\ref{sec-representer}, unregularized and at the rule
(\ref{eq-tuning}). Table~\ref{tbl-mc} reports median absolute deviation
from \(\beta_0\), which is the relevant summary because the
distributions are heavy-tailed at small \(\rho\).

\begin{longtable}[]{@{}
  >{\raggedright\arraybackslash}p{(\linewidth - 12\tabcolsep) * \real{0.0729}}
  >{\raggedleft\arraybackslash}p{(\linewidth - 12\tabcolsep) * \real{0.0521}}
  >{\raggedleft\arraybackslash}p{(\linewidth - 12\tabcolsep) * \real{0.1458}}
  >{\raggedleft\arraybackslash}p{(\linewidth - 12\tabcolsep) * \real{0.1458}}
  >{\raggedleft\arraybackslash}p{(\linewidth - 12\tabcolsep) * \real{0.0625}}
  >{\raggedleft\arraybackslash}p{(\linewidth - 12\tabcolsep) * \real{0.2500}}
  >{\raggedleft\arraybackslash}p{(\linewidth - 12\tabcolsep) * \real{0.2708}}@{}}

\caption{\label{tbl-mc}Median absolute deviation of \(\hat\beta_0\),
computed from the stored Monte Carlo output. The first column is the
infeasible estimator that observes every outcome, and is the noise
floor.}

\tabularnewline

\toprule\noalign{}
\begin{minipage}[b]{\linewidth}\raggedright
\(\rho\)
\end{minipage} & \begin{minipage}[b]{\linewidth}\raggedleft
\(n\)
\end{minipage} & \begin{minipage}[b]{\linewidth}\raggedleft
complete data
\end{minipage} & \begin{minipage}[b]{\linewidth}\raggedleft
complete case
\end{minipage} & \begin{minipage}[b]{\linewidth}\raggedleft
GMM
\end{minipage} & \begin{minipage}[b]{\linewidth}\raggedleft
representer \(\lambda=0\)
\end{minipage} & \begin{minipage}[b]{\linewidth}\raggedleft
representer \(\lambda^{*}\)
\end{minipage} \\
\midrule\noalign{}
\endhead
\bottomrule\noalign{}
\endlastfoot
0.05 & 500 & 0.030 & 0.497 & 0.423 & 0.733 & 0.464 \\
& 2000 & 0.015 & 0.498 & 0.362 & 0.635 & 0.440 \\
& 8000 & 0.008 & 0.500 & 0.276 & 0.317 & 0.271 \\
0.10 & 500 & 0.029 & 0.499 & 0.339 & 0.665 & 0.406 \\
& 2000 & 0.014 & 0.499 & 0.240 & 0.312 & 0.261 \\
& 8000 & 0.007 & 0.500 & 0.117 & 0.123 & 0.125 \\
0.20 & 500 & 0.027 & 0.504 & 0.231 & 0.295 & 0.257 \\
& 2000 & 0.016 & 0.499 & 0.132 & 0.132 & 0.124 \\
& 8000 & 0.007 & 0.500 & 0.066 & 0.064 & 0.063 \\
0.35 & 500 & 0.027 & 0.505 & 0.146 & 0.145 & 0.127 \\
& 2000 & 0.016 & 0.499 & 0.072 & 0.065 & 0.065 \\
& 8000 & 0.008 & 0.500 & 0.035 & 0.035 & 0.036 \\
0.50 & 500 & 0.030 & 0.501 & 0.088 & 0.098 & 0.094 \\
& 2000 & 0.014 & 0.500 & 0.053 & 0.049 & 0.049 \\
& 8000 & 0.008 & 0.499 & 0.026 & 0.025 & 0.025 \\

\end{longtable}

Three features deserve comment.

First, where the relevance and the sample are both small the
\emph{unregularized} representer estimator is worse than doing nothing:
its median absolute deviation exceeds that of complete-case
estimation---which is simply the bias \(0.5\) of the design---in 3 of
the 6 cells with \(\rho\le0.10\), most clearly at
\((\rho,n)=(0.05,500)\), where it is 0.733 against 0.497. Regularization
at (\ref{eq-tuning}) reverses it in every one of the 15 cells. The
advantage of the representer route over inverse weighting, which is an
advantage of order \(\rho^{-2}\) against \(\exp(v_\infty\gamma_Y^{2})\)
in population terms, does not survive estimation unless the conditioning
is repaired; by Theorem~\ref{thm-orders}
\(\operatorname{cond}(M)\asymp\rho^{-1}\), and that is what the ridge
addresses.

Second, at \(\rho\ge0.20\) all three feasible procedures are
indistinguishable. The choice between routes matters only in the regime
where neither works well, and there it matters for \emph{how} they fail:
the propensity route fails by non-convergence and exploding weights, the
representer route by silent heavy tails. The first is visible, the
second is not.

Third, the failure rate of the propensity route falls monotonically in
both arguments, from 0.248 at \((\rho,n)=(0.05,500)\) to 0.000 at
\(\rho\ge0.35\), \(n=8000\). The rate at the weakest cell --- roughly
one replication in four --- is informative and should be reported; the
corresponding replications are not missing at random and discarding them
silently biases any summary computed from the remainder.

Together with (\ref{eq-feas}), the picture is that weak relevance and
poor overlap degrade different things. Relevance controls the width of
the identified set and the conditioning of the sieve; overlap controls
feasibility and the tails. Neither substitutes for the other, and a
diagnostic addressed to one is silent about the other.

\subsection{Coverage}\label{sec-coverage}

Point estimates are not the object of the paper, so we also measure
coverage directly. At each candidate value \(s\) of the outcome
coefficient the four non-\(Z\) sample moments determine the remaining
parameters---the sample analogue of the curve of
Proposition~\ref{prp-curve}---and the remaining \(Z\)-moment is then a
scalar with population value \(\rho\varphi(s)\) by
Lemma~\ref{lem-signal}. Referring it to its estimated variance gives an
Anderson--Rubin statistic for \(H_0:\gamma_Y=s\), whose inversion is
identification-robust; we compare it with the Wald interval from the
just-identified generalized method of moments fit.
Table~\ref{tbl-coverage} reports 400 replications per cell, at fixed
\(\rho\) and along \(\rho_n=n^{-1/2}\).

\begin{longtable}[]{@{}lrllll@{}}

\caption{\label{tbl-coverage}Coverage of a nominal 95 per cent interval
for the outcome coefficient of the response model, the proportion of
scanned sets that come out bounded, and the proportion of grid points at
which the slice cannot be solved---a positivity or solver failure, not
an overflow of the tilt functional.}

\tabularnewline

\toprule\noalign{}
\(\rho\) & \(n\) & AR set & Wald interval & bounded & slice failure \\
\midrule\noalign{}
\endhead
\bottomrule\noalign{}
\endlastfoot
0.5 & 500 & 0.948 & 0.976 & 0.527 & 0.063 \\
0.5 & 2000 & 0.950 & 0.953 & 0.728 & 0.114 \\
0.5 & 8000 & 0.955 & 0.932 & 0.885 & 0.152 \\
0.2 & 500 & 0.963 & 0.930 & 0.155 & 0.063 \\
0.2 & 2000 & 0.955 & 0.920 & 0.338 & 0.114 \\
0.2 & 8000 & 0.953 & 0.946 & 0.585 & 0.152 \\
0.05 & 500 & 0.960 & 0.921 & 0.018 & 0.063 \\
0.05 & 2000 & 0.958 & 0.896 & 0.028 & 0.114 \\
0.05 & 8000 & 0.953 & 0.905 & 0.068 & 0.152 \\
\(n^{-1/2}\) & 500 & 0.960 & 0.921 & 0.018 & 0.063 \\
\(n^{-1/2}\) & 2000 & 0.958 & 0.890 & 0.010 & 0.114 \\
\(n^{-1/2}\) & 8000 & 0.950 & 0.909 & 0.020 & 0.152 \\

\end{longtable}

Three readings.

The Anderson--Rubin set holds its nominal level in every cell, from
\(\rho=0.5\) down to a sequence drifting to zero. That is what
identification-robust means and it is reassuring rather than surprising;
it is reported because the paper elsewhere declines to claim it for the
profiled likelihood of Section~\ref{sec-appl-set}, and the contrast
should be visible.

The Wald interval falls below its nominal level as relevance falls, and
under the drifting sequence it stays there: 0.921, 0.890, 0.909 at
\(n=500, 2000, 8000\). With 400 replications the standard error of each
entry is about 0.015, so we do not read a trend into those three
numbers; what matters is that a sixteenfold increase in sample size
produces no movement towards \(0.95\), while the Anderson--Rubin set is
at nominal throughout. That is the signature \citet{Dufour1997}
describes, and the reason a bounded interval cannot be valid in this
problem.

Reading the Wald column across rows rather than down them shows why that
failure is a statement about a \emph{sequence} and not about small
\(\rho\). At \(\rho=0.5\) it is within sampling error of nominal at
every sample size (0.976, 0.953, 0.932). At \(\rho=0.2\) it recovers:
0.930, 0.920, 0.946, back at nominal by \(n=8000\), which is what
fixed-\(\rho\) asymptotics require; Theorem~\ref{thm-if} is the paper's
own statement of that limit, for the representer route. At \(\rho=0.05\)
the same recovery has not happened by \(n=8000\) (0.921, 0.896, 0.905);
it is not denied, only postponed, because the sample size at which it
arrives grows as the relevance falls. Along the drifting sequence it
never arrives, because the relevance falls as fast as the sample size
grows. That is the non-uniformity Theorem~\ref{thm-drifting} establishes
for the representer route, exhibited here by the propensity route in one
column: every fixed relevance is eventually enough, and no sequence
approaching zero ever is. A reader who notices the 0.946 at \(\rho=0.2\)
and takes it for a counterexample has read the table down instead of
across.

The third column is the one this paper is about, and the contrast in it
is sharp. At every fixed relevance the proportion of scanned sets that
come out bounded rises steeply with \(n\)---from 0.527 to 0.885 at
\(\rho=0.5\), and from 0.018 to 0.068 even at \(\rho=0.05\)---because
identification improves. Under the drifting sequence it does not: 0.018,
0.010, 0.020 over the same sixteenfold range. The scanned set is
therefore honestly unbounded exactly where the data remain uninformative
however large the sample, which is the behavior one gets when the
statistic being inverted is the right one.

The fourth column requires care, and an earlier version of this paper
did not take it. It records the proportion of grid points at which the
slice cannot be solved: the positivity floor is breached, or the inner
solver fails. That proportion rises with \(n\), from 6.3\% at \(n=500\)
to 15.2\% at \(n=8000\), which is worth reporting in its own right---a
larger sample draws more extreme outcomes, so at a fixed \(\gamma_Y\)
the fitted response probabilities breach the floor more often. But it is
\textbf{not} evidence for Theorem~\ref{thm-reach}. The mechanism there
is the growth of \(\Lambda\) and the collapse of the effective sample
size, neither of which this column tests; and the direction is wrong,
since (\ref{eq-reach}) has the reach \emph{growing} in \(n\) on a grid
that is held fixed here. Two further remarks. The shadow variable enters
only the \(Z\)-moment and not the system that determines the remaining
parameters, so this rate is identical across \(\rho\) by construction
rather than by accident; it confirms the structure of
Proposition~\ref{prp-rho-free} without being independent evidence for
it. And the grid used here, \([-6,4]\), comes nowhere near the ceiling
\(|\gamma_Y|\le18.8\) of Theorem~\ref{thm-reach}, so this design could
not have exhibited the ceiling in any case.

\subsection{The reach, measured directly}\label{sec-reach-sim}

Theorem~\ref{thm-reach} is tested separately, and on the quantity it is
about. At scan point \(s\) the working propensity is computable for
every respondent, because \(Y\) is observed there and
Lemma~\ref{lem-intercept} pins the intercept in closed form; the
effective sample size follows, and the reach is the largest \(|s|\) at
which it still exceeds \(m=30\).

\begin{longtable}[]{@{}rrcr@{}}

\caption{\label{tbl-reachsim}The measured reach of the scan against
(\ref{eq-reach}), over a sixty-four-fold range of sample size; 200
replications per row, median and interquartile range.}

\tabularnewline

\toprule\noalign{}
\(n\) & measured reach & interquartile range & (\ref{eq-reach}) \\
\midrule\noalign{}
\endhead
\bottomrule\noalign{}
\endlastfoot
2,000 & 2.46 & {[}2.10, 3.02{]} & 1.62 \\
8,000 & 2.84 & {[}2.42, 3.48{]} & 1.82 \\
32,000 & 3.14 & {[}2.72, 3.70{]} & 2.00 \\
128,000 & 3.48 & {[}3.00, 4.11{]} & 2.17 \\

\end{longtable}

Table~\ref{tbl-reachsim} shows the functional form
Theorem~\ref{thm-reach} claims. Over the sixty-four-fold range the reach
grows by a factor of 1.41, against 1.34 for the square-root-of-logarithm
law of (\ref{eq-reach}) and 1.80 for a reach growing like the logarithm
itself. Four points cannot do more than discriminate between candidates
of that separation, and they do: a thousandfold increase in sample size
buys about one unit of \(\gamma_Y\), which is the substantive content of
the theorem and what makes the scan unable to keep up with a widening
set.

What is not recovered is the slope, and the discrepancy is instructive.
The fitted \[
  |\gamma_Y|_{\max}^{2}
  =-1.53+1.44
   \left[\log(n/m)-2\log\operatorname{pr}(R=0)\right]
\] has slope 1.44 against the predicted \(1/v_\infty=0.50\), so the
measured reach exceeds (\ref{eq-reach}) throughout and by a widening
margin: the ratio runs 1.52, 1.56, 1.57, 1.61 over the four sample
sizes. A margin that grows rules out a pure level effect. The reason is
the one met in Section~\ref{sec-decomposition}, sharpened: \(\Lambda\)
is an expectation over a tail that no finite sample contains, so
\(\log\hat\Lambda\) is biased downwards---and the bias grows with the
depth of the scan, because the further the tilt is carried the smaller
the fraction of its mass the sample represents. A downward bias in
\(\log\hat\Lambda\) that increases in \(|\gamma_Y|\) is exactly a slope
deficit in the fitted reach. The practical reading is the uncomfortable
one: the empirical effective sample size \emph{overstates} how far the
scan can be carried, and overstates it more the further one goes. An
analyst monitoring it would place the arithmetic boundary beyond where
it lies, and would be misled in the direction that matters.

\subsection{The screen when the outcome is
continuous}\label{sec-contscreen}

Item (i) of Section~\ref{sec-diagnostics} prescribes a \(\chi^{2}\) per
degree of freedom. The design of this section lets that prescription be
checked rather than argued, because it supplies what an application
cannot: the relevance \(\rho\) is known, and \emph{exclusion holds
exactly}, since \(R\) depends on \((X,Y)\) alone and \(Z\) enters only
through the residual of \(Y\). Any exclusion signal the screen reports
here is therefore contamination and not violation, and its size can be
measured. The outcome is coded into \(k\) classes of equal mass and the
outcome model fitted as a cumulative logit, as item (i) describes; the
response model is unchanged, and needs no coding, so the exclusion axis
does not depend on \(k\) at all.

Three things follow, from 200 replications at \(n=4000\).

First, the scale is calibrated. At \(\rho=0\) the shadow variable is
independent of the outcome given \(X\), and the relevance axis sits at
its own noise floor: the median \(\chi^{2}\) per degree of freedom runs
0.41 to 0.53 across the four codings, against the median 0.45 of the
\(\chi^{2}_{1}\) an absent effect produces. The standard of \(10\) is
therefore a demand that the signal exceed by tenfold what the same
coding yields from nothing at all.

Second, the warning in item (i) is not hypothetical. At \(\rho=0.35\)
the incremental \(R^{2}\) falls from 0.0340 to 0.0179 as \(k\) goes from
3 to 20, while the \(\chi^{2}\) it is computed from \emph{rises} from
120.8 to 175.1. The measure moves against the signal, and only the ratio
to the degrees of freedom is comparable.

Third, and least comfortable: a validly excluded shadow variable does
not sit far below the equality line. With exclusion holding by
construction, the exclusion axis still runs from 0.42 at \(\rho=0\) to
157.4 at \(\rho=0.50\), and the ratio of the two axes stays between 1.08
and 2.50 while relevance itself moves over a factor of 42. The
contamination is proportional to the relevance, because the only route
from \(Z\) to \(R\) is through \(Y\) and that route widens exactly as
\(Z\) predicts \(Y\) better. Two consequences for reading
Figure~\ref{fig-frontier}. The ratio is not a test of exclusion in
either direction, which is what Section~\ref{sec-screens} says of it.
And the emptiness of the region a usable shadow variable would occupy is
\emph{necessary} for the conclusion there rather than sufficient: a
strong dependence of response on the outcome pulls even a valid
candidate towards the line, so a screen run on data with a weaker
\(\gamma_Y\) would place a valid candidate further below it than this
one does. What carries the conclusion of Section~\ref{sec-screens} is
that exclusion is rejected for every characteristic and that relevance
is small in absolute terms, not the position of any point relative to
the diagonal.

\section{Application: reject inference in consumer
credit}\label{sec-application}

\subsection{The data, and a reproduction}\label{sec-data}

A German bank recorded 9,780 consumer-credit applications, of which
4,000 were accepted; repayment is observed only for those, and 169
defaulted. The acceptance decision was made on a scorecard built from 13
characteristics, and the missingness is nonignorable by construction:
the bank declined the applicants it judged most likely to default. These
data were analyzed by \citet{BueckervanKampenKraemer2013} using the
semiparametric empirical likelihood of \citet{QinLeungShao2002}; the
present author is one of its authors. Doubts about whether reject
inference delivers what it promises are not new, and have recently been
given an empirical form by \citet{ScaroneBaezaYates2026}, who find that
in a retraining cycle such methods can raise accuracy while recall
collapses. The diagnosis below is of a different kind---not that the
correction fails to help, but that the data do not identify what it
corrects for.

We first reproduce it. The preparation conventions were recovered from
the original analysis scripts, three of which could not have been
inferred from the published text and are documented in
Section~\ref{sec-data-prep}. On that basis all forty-one coefficients of
the published acceptance model agree to the three printed decimal
places, as do their standard errors and \(p\)-values, and forty of the
forty-one coefficients of the published semiparametric estimate; the
exception differs by \(0.001\). The reproduction also recovers two
errors in the original data preparation, one in a recoding of age and
one in a duration variable whose factor levels differ between the
accepted and rejected files; neither alters the conclusions below, and
both are reported in Section~\ref{sec-data-prep} so that the published
tables can be reconciled.

The purpose of the reproduction is not archival. Only with the
specification pinned down exactly can the diagnostics of
Section~\ref{sec-diagnostics} be applied to the analysis as it was
actually performed.

\subsection{The shadow variable, and a screen}\label{sec-screens}

The order condition requires one characteristic to be excluded from the
response model, and the published analysis states which: months at the
current address. The choice was defensible on the evidence then
considered, that variable being the least significant predictor of
acceptance, with \(p=0.939\). Our reproduction returns exactly that
value, which is a useful check that the model being diagnosed is the
model that was fitted.

Choosing the excluded characteristic by its insignificance in the
response model is not an idiosyncrasy of this analysis; it is what a
formal selection procedure also does. \citet{ChenShaoFang2021} and
\citet{ChenShao2025} select an instrument by a criterion that detects
departures from the exclusion restriction, and \citet{ZhaoWangShao2021}
search for one among linear combinations of the covariates. All three
would have ranked this characteristic highly, because on the exclusion
axis it is the best candidate available. The screen below adds the
second axis.

Its relevance was not examined. It is negligible: the characteristic
contributes \(\chi^{2}=0.33\) on 1 degree of freedom (\(p=0.57\)) to the
outcome model among respondents, a first-stage analogue of 0.33 against
a standard of \(10\).

The effective sample size tells the same story, provided it is evaluated
at the estimate rather than at ignorability. Weighting by the fitted
\emph{ignorable} propensity---the model without the outcome---gives
\(\mathrm{ESS}=342\), or 0.035 of the 9,780 applications. But the
published analysis estimates \(\gamma_Y=6.845\), and at that value the
working propensity is computable for every respondent, giving
\(\mathrm{ESS}=154\), or 0.016. The nonignorable model the analysis
actually fits leaves 0.45 of the information the ignorable one would
suggest: it more than halves it. Throughout this section effective
sample sizes are quoted as a fraction of applications, for which
Corollary~\ref{cor-ess-bound} gives the ceiling
\(\operatorname{pr}(R=1)=0.409\).

Screening all 13 characteristics shows this was not an unlucky choice.
The 4 for which exclusion is not rejected have relevance of at most 3.0
per degree of freedom, against the standard of \(10\) used above; the
only characteristic that clears that standard is the credit bureau
score, at 16.8, and it ranks 3 of 13 on the exclusion axis as well.
Because the sample here is small enough for tests to be meaningful, and
large samples require effect sizes instead, we measure both axes by
incremental McFadden \(R^{2}\): relevance for \(Y\) among the accepted
against the same quantity for \(R\) on the pooled sample.
Figure~\ref{fig-frontier} shows the result.

The same screen on a second and much larger data set gives the same
answer more sharply. Lending Club publishes its declined applications
alongside its funded ones; for the 2014 vintage, 1,839,397 applications
of which 223,102 were funded and 41,161 defaulted, with the credit score
observed on both sides. Here too nothing occupies the region a usable
shadow variable would have to occupy. Relevance is small
throughout---the largest, the credit score's, explains 1.71 per cent of
the null deviance---and the score, the most relevant of the 6
characteristics, has a ratio of 0.19: it enters the acceptance decision
some 5 times as strongly as it predicts repayment. There are 2
characteristics whose relevance does exceed their exclusion violation,
but that is no licence to use them: exclusion is rejected for all 6, the
weakest violation still contributing 12 per degree of freedom to the
response model, and a ratio above one says only that the violation is
smaller than the relevance, not that it is absent.
Section~\ref{sec-selection} explains why this had to happen.

The coding of the screen matters here, and it is worth saying why.
Lending Club applies a hard floor on the credit score: in the 2014
vintage the lowest funded score is 662, the first percentile of the
funded scores is the same, and 66 per cent of the declined applications
lie below it. A bare linear term cannot represent a step, so a screen
that enters the score linearly understates how much of the decision it
explains: coded that way its ratio comes out at 0.95 rather than 0.19,
which is the difference between a candidate that looks borderline and
one that is plainly unusable. The screen therefore codes them in
bands---here, and not in the first panel, where neither condition for
doing so holds: the decision must plausibly use the characteristic
through a threshold, and the sample must carry the degrees of freedom
that bands cost. With 9,780 applications it does not. Banding there
would take the outcome model from 41 parameters to 67 on 169
minority-class cases and multiply the measured relevance of the least
relevant characteristic 17-fold, with nothing in the data changing. The
threshold is the general point: wherever a decision uses one, a linear
screen flatters the characteristic as a shadow variable.

\begin{figure}

\centering{

\includegraphics[width=1\linewidth,height=\textheight,keepaspectratio]{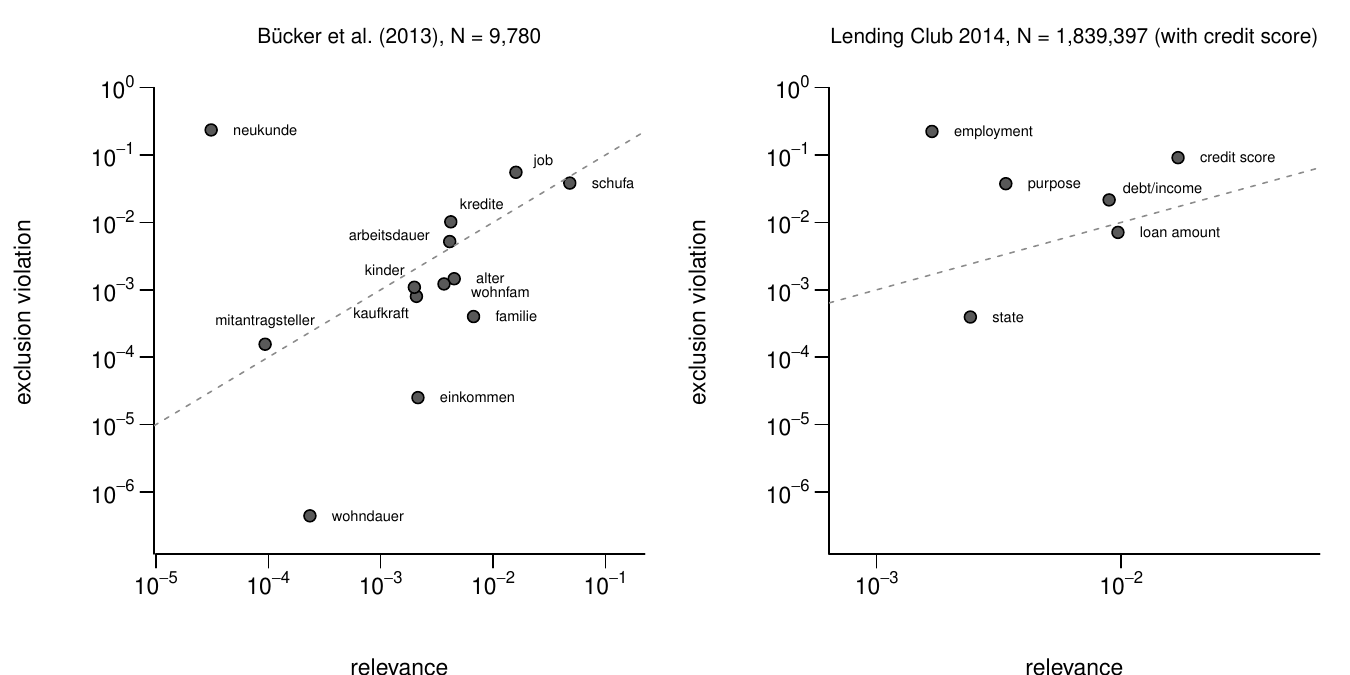}

}

\caption{\label{fig-frontier}Relevance against exclusion violation for
every characteristic, in the two data sets. Both axes are incremental
McFadden \(R^{2}\): for the outcome among respondents on the horizontal
axis, for the response indicator on the pooled sample on the vertical.
Both are logarithmic and share a common vertical scale across panels, so
the broken line marks equality. It is not at \(45^{\circ}\) on the page,
because the two axes span different numbers of decades. A usable shadow
variable would need substantial relevance together with a small
exclusion violation: far to the right on the horizontal axis and well
below the line. The characteristics that do lie well below the line are
the weakly relevant ones, so that region is empty in both panels---which
is necessary for the conclusion drawn from it rather than sufficient,
since a response probability depending strongly on the outcome pulls
even a validly excluded candidate towards the line. In each panel the
most relevant characteristic is the credit bureau score: in the bank
data it lies essentially on the line, in the Lending Club data well
above it.}

\end{figure}%

\subsection{What the profiled likelihood shows}\label{sec-appl-set}

The published estimate of the outcome coefficient in the response model
is \(6.845\) with standard error \(2.987\), an interval of
\([0.99,12.70]\).

Before reporting an alternative we must be explicit about what it is. We
profile the semiparametric likelihood over the outcome coefficient,
holding it fixed and maximizing over the remaining forty coefficients
and the response rate, and calibrate with \(\chi^{2}_{1}\) quantiles.
\textbf{This is not an identification-robust procedure.} Under weak
identification the profile likelihood ratio for a subvector is in
general not asymptotically \(\chi^{2}_{1}\); that is precisely the
separation between Anderson--Rubin-type statistics, which remain valid,
and likelihood-ratio, Lagrange multiplier and Wald statistics, which do
not \citep{StockWright2000, GuggenbergerSmith2005}.
\citet{NakanoHoshino2026} invert an Anderson--Rubin statistic for this
reason. We have not established the coverage of the sets below, and we
do not claim it.

What survives the objection is the \emph{shape} of the profile rather
than its level, and it is the shape that carries the argument of this
section: a profile that converges to a finite asymptote bounds the
evidence against extreme values no matter how it is calibrated, and
whether the reported set closes then depends on the nominal level rather
than on the data. With that caveat, the profile gives the nominal
\(90\%\) set \([1.82,12.10]\), the \(95\%\) set \([0.46,13.35]\), and at
\(99\%\) a half-line bounded above by \(16.46\) only.

Near the maximum the profile is quadratic and symmetric with precisely
the curvature belonging to the reported standard error: at distances
\(0.5,1.0,1.5,2.0\) the likelihood ratio is \(0.027,0.110,0.247,0.440\)
against the quadratic values \(0.028,0.112,0.252,0.448\), and the two
sides of the maximum differ by at most 0.004. The standard error is not
miscomputed, and at the conventional levels the substantive conclusion
survives: acceptance depends on creditworthiness after conditioning on
the score characteristics.

The profile is not, however, globally quadratic. Descending, it
converges monotonically to a finite asymptote at a likelihood ratio of
6.2554, flat from about \(-11\) down to \(-53\); the outcome is binary,
so driving the coefficient to \(-\infty\) with a compensating intercept
approaches the boundary model in which no good risk is ever accepted,
which has finite likelihood. A search over forty-three starting points
at two points in that tail, including a systematic exploitation of the
compensation structure, found nothing better. Because 6.2554 lies below
the \(99\%\) critical value of 6.635, \emph{the \(99\%\) set is a
half-line rather than an interval}: the lower bound of the reported set
is a property of the nominal level, not of the data.
Figure~\ref{fig-profile} shows it.

\begin{figure}

\centering{

\includegraphics[width=0.75\linewidth,height=\textheight,keepaspectratio]{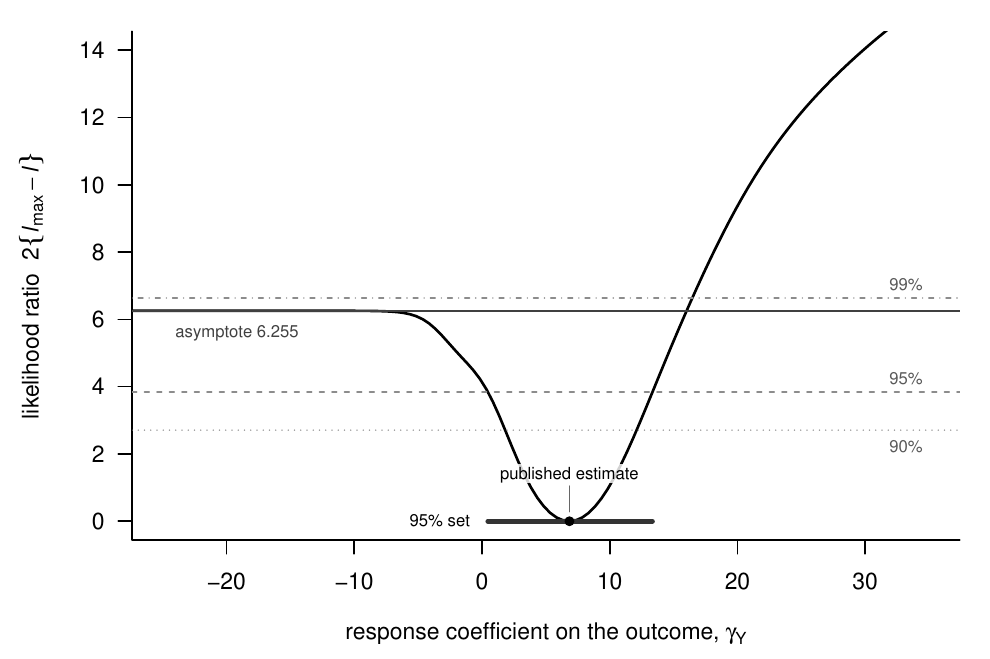}

}

\caption{\label{fig-profile}Profile likelihood ratio for the outcome
coefficient of the response model, over 237 points. Near the maximum the
profile is quadratic with the curvature implied by the reported standard
error. Descending, it converges to a finite asymptote, which lies below
the \(99\%\) critical value: the \(99\%\) set is a half-line. The bar
marks the \(95\%\) set.}

\end{figure}%

On the odds scale the point estimate says the odds of acceptance are
\(e^{6.845}\approx939\) times larger for a good risk. The \(95\%\)
robust set runs from \(e^{0.46}\approx1.6\) to
\(e^{13.35}\approx6.3\times10^{5}\), from barely selective to
effectively deterministic. That range, and not the standard error, is
what these data determine about the response mechanism, and it is what
one should expect from a shadow variable contributing
\(\chi^{2}/\mathrm{df}=0.33\).

\subsection{A positivity failure that precedes all of
this}\label{a-positivity-failure-that-precedes-all-of-this}

One diagnostic should be applied before any of the above. New customers
are accepted at 6.5\% against 63.4\% for existing customers. On the
pooled sample the smallest fitted acceptance probability among the
accepted is 0.0033, the largest inverse weight is 307, and the effective
sample size is 0.035 of applications against the ceiling 0.409 of
Corollary~\ref{cor-ess-bound}---a shortfall of a factor of 12.
Restricting to existing customers raises it to 0.330 of 5,908 against a
ceiling of 0.634, and the largest weight falls from 307 to 50.

The restriction is not a remedy. It discards the 3,872 new customers,
who are the applicants about whom reject inference is asked in the first
place; what it buys in overlap it spends on the estimand. The
reproduction above therefore follows the original in using the pooled
sample, and the diagnostic is reported as a property of that analysis
rather than as a correction applied to it. The failure is invisible in a
coefficient table and fatal to any inverse-weighted method; it is item
(iii) of Section~\ref{sec-diagnostics}.

Two failures have now been recorded in these data: the weighting cannot
be repaired within them, and the shadow variable that would have made
the nonignorable model estimable is not among the characteristics
recorded. The first is a property of this lender's acceptance rate. The
second is not, and Section~\ref{sec-selection} shows why.

\subsection{What the scorecard's own coefficients are identified to
be}\label{sec-scorecard}

Section~\ref{sec-screens} leaves the lender's own data with no usable
shadow variable among the 13 characteristics recorded: the 4 whose
exclusion survives its test carry relevance of at most 3.0 per degree of
freedom, far below the standard used there, and the one characteristic
that does clear that standard, the credit bureau score, itself ranks 3
of 13 on the exclusion axis. That settles what cannot be done with these
data but not what can. A reader who grants instead that the departure
from missingness at random is \emph{bounded} will ask the next question:
under such a bound, what are the coefficients of the scorecard itself
identified to be? That question has an answer, and this subsection
reports it for the 9 coefficients that the accepted-only fit makes
significant at the 5\% level on the 5,908 existing customers.

Let \(\vartheta\ge 0\) bound the tilt of
\(\operatorname{pr}(Y=1\mid x)\) away from the respondents' conditional
mean, so that within the same recorded characteristics a good
application is at most \(e^{\vartheta}\) times as likely to be accepted
as a bad one, and intersect that bound with the assumption-free one.
Each applicant's conditional mean is then known only to lie in an
interval, and the scorecard is the logistic projection of those means on
the model's own design. The identified set for the coefficient vector is
therefore the \emph{image} of a box under a nonlinear map, which need
not be convex. A bound on one coordinate \(\beta_j\) is still attained
at a corner of the box --- at the upper endpoint on the applicants where
\(W^{\mathrm T}H^{-1}e_j>0\) and at the lower endpoint where it is
negative --- at neither, on the applicants where it vanishes, since
\(m\) is then unconstrained there --- with
\(H=E\{WW^{\mathrm T}g'(W^{\mathrm T}\beta)\}\). But \(H\) depends on
the solution, so the selection is a fixed point rather than a formula
and may have several local solutions.

\textbf{Nothing in that construction is new, and the citations matter
here more than usual.} The bound is the standard sensitivity class for
nonignorable selection, a parameterized selection-bias function tilting
the respondents' law
\citep{ScharfsteinRotnitzkyRobins1999, RobinsRotnitzkyScharfstein2000},
here in the bounded-odds form whose sharp bounds \citet{DornGuo2023}
characterize for inverse-propensity weighting; it sits in the
partial-identification tradition of \citet{Manski1989} and
\citet{Manski2016} and beside the bounded selection probabilities of
\citet{AronowLee2013}, while \citet{KlineSantos2013} index the same
departure differently, by a Kolmogorov--Smirnov distance rather than an
odds bound. Reporting an identified set and a confidence set side by
side, as Table~\ref{tbl-scorecard} does, is the ignorance and
uncertainty region distinction of \citet{VansteelandtEtAl2006}. The sign
rule and the simultaneous shift of all units are the imputation towards
directional extremes of \citet{VansteelandtGoetghebeur2001}: their
equation (3.3) is this rule for a generalized linear model, their §3.1
recommends moving every unit at once on the sign, and their pp.~661--662
recommend randomly chosen starting imputations, among other devices,
where global convergence is in doubt. What differs is the set the
imputations range over --- a box from a bounded tilt in place of their
imputation horizons --- and the inference: they take the confidence
limits computed from the two extreme completed datasets, where the
endpoint here has an influence function of its own. The
threshold-indicator fixed point corresponds to that of
\citet{BonviniEtAl2022} for marginal structural models, who also record
that outside a linear model the selection is intractable and solved only
locally. The numbers below are an application; the method is theirs.

\textbf{At \(\vartheta=0\) the set is a point, and that point is the
ordinary coefficient --- but not with the ordinary standard error.} The
box collapses, and because the fitted values of the accepted-only
logistic regression are what the box collapses to, the projection
returns that regression's coefficients exactly, to the last digit. What
changes is the uncertainty: the target is the projection over \emph{all}
applicants, including the rejected ones about which the data say less,
and its standard error is up to 1.87 times the accepted-only one. That
is not a detail. Of the 9 coefficients, 3 have 95\% sets containing zero
already at \(\vartheta=0\) --- their accepted-only statistics run from
2.30 to 2.73 in absolute value and their projected ones from 1.44 to
1.70. The estimate did not move; the question did. Reading the
accepted-only significances as though they were significances for this
target is the easiest mistake to make with Table~\ref{tbl-scorecard},
and it would overstate every conclusion drawn from it.

\begin{longtable}[]{@{}
  >{\raggedright\arraybackslash}p{(\linewidth - 12\tabcolsep) * \real{0.2027}}
  >{\raggedleft\arraybackslash}p{(\linewidth - 12\tabcolsep) * \real{0.0946}}
  >{\raggedleft\arraybackslash}p{(\linewidth - 12\tabcolsep) * \real{0.1554}}
  >{\raggedleft\arraybackslash}p{(\linewidth - 12\tabcolsep) * \real{0.1892}}
  >{\raggedleft\arraybackslash}p{(\linewidth - 12\tabcolsep) * \real{0.1149}}
  >{\raggedleft\arraybackslash}p{(\linewidth - 12\tabcolsep) * \real{0.1216}}
  >{\raggedleft\arraybackslash}p{(\linewidth - 12\tabcolsep) * \real{0.1216}}@{}}

\caption{\label{tbl-scorecard}Identified sets for the scorecard
coefficients under a bounded departure from missingness at random, on
the existing-customer stratum. The first column is the point the set
collapses to at \(\vartheta=0\), which is the coefficient of the
accepted-only logistic regression exactly; the next two are the
estimated set and the 95 per cent confidence set at the largest tilt
considered. The last three columns give the smallest tilt on the grid at
which the set, and then each confidence set, contains zero: a dash means
it does not by the end of the grid. Rows are ordered by how long the
coefficient's sign survives.}

\tabularnewline

\toprule\noalign{}
\begin{minipage}[b]{\linewidth}\raggedright
coefficient
\end{minipage} & \begin{minipage}[b]{\linewidth}\raggedleft
\(\vartheta=0\)
\end{minipage} & \begin{minipage}[b]{\linewidth}\raggedleft
set at \(\vartheta=0.5\)
\end{minipage} & \begin{minipage}[b]{\linewidth}\raggedleft
95\% set at \(\vartheta=0.5\)
\end{minipage} & \begin{minipage}[b]{\linewidth}\raggedleft
\(\vartheta\): set
\end{minipage} & \begin{minipage}[b]{\linewidth}\raggedleft
\(\vartheta\): 90\%
\end{minipage} & \begin{minipage}[b]{\linewidth}\raggedleft
\(\vartheta\): 95\%
\end{minipage} \\
\midrule\noalign{}
\endhead
\bottomrule\noalign{}
\endlastfoot
credit bureau class P & \(-4.05\) & \([-5.31, -2.78]\) &
\([-6.88, -1.70]\) & --- & --- & --- \\
occupation: self-employed & \(-2.28\) & \([-3.83, -0.68]\) &
\([-5.19, 0.70]\) & --- & 0.400 & 0.300 \\
occupation: skilled worker & \(-1.67\) & \([-2.83, -0.47]\) &
\([-4.12, 0.76]\) & --- & 0.250 & 0.200 \\
credit bureau class K--O & \(-1.33\) & \([-2.43, -0.07]\) &
\([-3.59, 0.72]\) & --- & 0.250 & 0.200 \\
credit bureau class F--J & \(-0.83\) & \([-1.75, 0.12]\) &
\([-2.42, 0.65]\) & 0.500 & 0.200 & 0.150 \\
occupation: salaried employee & \(-1.21\) & \([-2.18, -0.19]\) &
\([-3.36, 0.83]\) & --- & 0.150 & 0.075 \\
employment duration & \(0.004\) & \([-0.002, 0.009]\) &
\([-0.006, 0.014]\) & 0.400 & 0.000 & 0.000 \\
number of loans: other & \(-0.85\) & \([-2.06, 0.32]\) &
\([-3.01, 1.68]\) & 0.400 & 0.000 & 0.000 \\
age & \(-0.03\) & \([-0.08, 0.02]\) & \([-0.11, 0.07]\) & 0.300 & 0.025
& 0.000 \\

\end{longtable}

\textbf{One characteristic survives the whole range.} Taking the sets as
estimated, 5 of the 9 coefficients keep their sign to \(\vartheta=0.5\);
taking the 95\% confidence sets, 1 does. It is class P, the lowest of
the ordered credit-bureau classes, whose set at \(\vartheta=0.5\) is
\([-5.31, -2.78]\) and whose 95\% confidence set is \([-6.88, -1.70]\),
against \(-4.05\) under missingness at random. A bound of
\(\vartheta=0.5\) is not a weak assumption --- it permits an acceptance
probability ratio of 1.65 within identical recorded characteristics ---
and under it the sign of that one coefficient is all that the data still
deliver.

\textbf{What this inference assumes, and where it has not been checked.}
The confidence sets add a one-sided normal interval to each endpoint
separately, so they are conservative in the same way a Bonferroni bound
is. Their validity rests on the maximizer being unique and on a margin
condition for the direction weights. The vanishing case above is not
hypothetical, but it is small here: the corner rule left \(m\) free on
at most 6 of the 5,908 applicants at any one endpoint, and 19 of the 198
endpoints needed the smooth maximiser rather than corner selection
alone. Uniqueness is supported here rather than assumed --- each of the
198 endpoint computations was restarted from 8 selections, including
random half-spaces, and the best restart improved on the primary
solution by less than \(10^{-4}\) of its standard error --- but the
margin condition is not verified, the nuisances are parametric logistic
fits that are not cross-fitted, and the asymptotic coverage of these
intervals has been checked numerically only in a design where the true
nuisances are known, not on these data. Table~\ref{tbl-scorecard} is
reported as what it is: a sensitivity analysis whose point estimates are
firmer than its confidence sets, and whose confidence sets are the ones
to read.

\section{Relevance and exclusion under an optimizing
selector}\label{sec-selection}

The results of this section concern selection mechanisms in general
rather than credit data. They are placed after
Section~\ref{sec-application} because those data are what prompted them,
and because they account for what the screens there found.

\subsection{Why no shadow variable was available}\label{sec-noshadow}

The screens of Section~\ref{sec-screens} found no usable shadow variable
in either data set: in the first, the four characteristics for which the
exclusion restriction is not rejected have negligible relevance; in the
second, it is rejected for every characteristic, and the single most
relevant one---the credit score---violates it some 5 times more strongly
than it predicts the outcome. That is not accidental.

\begin{proposition}[Selection by an optimizing
scorer]\protect\hypertarget{prp-selection}{}\label{prp-selection}

Suppose the response is generated by an optimizing selector:
\(R=\eta\{f^{*}(V),U\}\) with characteristics \(V=(V_1,\dots,V_p)\),
\(U\) independent of \((V,Y)\), \(\eta\) strictly monotone in its first
argument, and \(f^{*}\in\arg\min_{f\in\mathcal F}E[L\{Y,f(V)\}]\) for a
strictly proper loss \(L\) and a class \(\mathcal F\). Write
\(\mathcal F_{-j}\subset\mathcal F\) for the elements not depending on
\(v_j\), and assume \(\mathcal F\) closed under almost-sure equality, so
that any \(f\in\mathcal F\) agreeing \(P_V\)-almost surely with a
function of \(v_{-j}\) alone has that function in \(\mathcal F_{-j}\).
If \(V_j\) satisfies the exclusion restriction
\(V_j\perp\!\!\!\perp R\mid(Y,V_{-j})\), then \(f^{*}\) agrees
\(P_V\)-almost surely with an element of \(\mathcal F_{-j}\), and
consequently \[
  \min_{f\in\mathcal F}E[L\{Y,f(V)\}]=\min_{f\in\mathcal F_{-j}}E[L\{Y,f(V_{-j})\}] :
\] the characteristic buys no reduction in risk within \(\mathcal F\).

\end{proposition}

\begin{proof}
Because \(U\) is independent of \((V,Y)\) and \(R\) depends on \(V\)
only through \(f^{*}(V)\),
\(\operatorname{pr}(R=1\mid Y,V)=g\{f^{*}(V)\}\) for
\(g(t)=E_U\{\eta(t,U)\}\), which is strictly monotone. Exclusion states
that this conditional probability does not depend on \(V_j\) given
\(V_{-j}\); strict monotonicity of \(g\) then forces \(f^{*}(V)\) to be
constant in \(V_j\) for \(P_{V_{-j}}\)-almost every \(V_{-j}\), so
\(f^{*}\) agrees \(P_V\)-almost surely with a function of \(V_{-j}\)
alone, which lies in \(\mathcal F_{-j}\) by the closure assumption. The
displayed identity follows, since the minimum over \(\mathcal F\) is
then attained in \(\mathcal F_{-j}\). Strict properness of \(L\) enters
only through the existence and uniqueness of \(f^{*}\); the argument
above uses the other three alignments once each.
\end{proof}

Both qualifications are needed and neither is vacuous. The conclusion is
almost sure with respect to \(P_V\), so it constrains \(f^{*}\) only
where \(V_j\) actually varies given \(V_{-j}\): if the conditional
support of \(V_j\) given \(V_{-j}\) is a single point almost surely, the
proposition holds but says nothing, because no function could have
depended on \(V_j\) in any case. It has content exactly when \(V_j\)
carries conditional variation, which is the situation of interest---a
characteristic recorded on the application form because it differs
across applicants with otherwise similar profiles. The closure
assumption rules out classes contrived so that the only representative
of a \(v_{-j}\)-measurable function carries a redundant \(v_j\)
argument; it holds for the linear and additive classes used in
scorecards, and for any class defined by a constraint on the function
rather than on its arguments.

The mechanism is worth stating plainly. The selector reaches \(V_j\)
only through the score, so if the score moved with \(V_j\) the
acceptance probability would move with it, and that is precisely what
exclusion forbids. Exclusion therefore does not merely restrict how
\(V_j\) may enter the response mechanism; it requires that the score not
use \(V_j\) at all.

The contrapositive is the statement of practical interest: \emph{a
characteristic that is relevant for the outcome within \(\mathcal F\)
cannot satisfy the exclusion restriction.} Relevance and exclusion are
not in tension by accident of a particular data set; under an aligned
selector they are incompatible, and a scorecard built to predict default
is precisely such a selector.

\subsection{How much misalignment is needed}\label{sec-misalignment}

Proposition~\ref{prp-selection} says that an aligned selector leaves no
relevant characteristic excluded. The practical question is the converse
one: how far must the selector's objective depart from the analyst's
before a usable shadow variable appears? The first half of the answer is
immediate.

\begin{proposition}[Exclusion and invisibility are the same
thing]\protect\hypertarget{prp-converse}{}\label{prp-converse}

Under the conditions of Proposition~\ref{prp-selection}, \(V_j\)
satisfies the exclusion restriction if and only if \(f^{*}\) does not
depend on \(V_j\).

\end{proposition}

\begin{proof}
Sufficiency is Proposition~\ref{prp-selection}. For necessity,
\(\operatorname{pr}(R=1\mid Y,V)=g\{f^{*}(V)\}\) with \(g\) strictly
monotone, so if \(f^{*}\) depends on \(V_j\) given \(V_{-j}\) on a set
of positive probability then so does the response probability, and
exclusion fails.
\end{proof}

Exclusion is therefore not a matter of degree in the selector's
objective: it is invisibility to it. What remains is quantitative, and
it is not what one would guess. Write \(\mathrm{exc}_j\) for the risk
reduction \(V_j\) buys for \(R\) within \(\mathcal F\)---the exclusion
violation of Section~\ref{sec-screens}---and \(\mathrm{rel}_j\) for the
risk reduction it buys for \(Y\) in the richer class in which relevance
is measured. A usable shadow variable needs
\(\mathrm{rel}_j/\mathrm{exc}_j\) large. Table~\ref{tbl-misalignment}
tilts a characteristic's contribution to the outcome continuously out of
the scorer's class, from full alignment to none, holding its variance
fixed.

\begin{longtable}[]{@{}rrrrl@{}}

\caption{\label{tbl-misalignment}Relevance against exclusion violation
as a characteristic's contribution to the outcome is tilted out of the
scorer's class. The alignment \(\phi=1\) places it entirely inside;
\(\phi=0\) entirely outside. There is no intermediate regime.}

\tabularnewline

\toprule\noalign{}
\(\phi\) & relevance & exclusion violation & ratio & form that helps \\
\midrule\noalign{}
\endhead
\bottomrule\noalign{}
\endlastfoot
1.0 & 0.0640 & 0.0477 & 1.34 & linear \\
0.8 & 0.0459 & 0.0338 & 1.36 & linear \\
0.6 & 0.0255 & 0.0193 & 1.32 & linear \\
0.5 & 0.0185 & 0.0142 & 1.31 & linear \\
0.5 & 0.0172 & 0.0001 & 285 & quadratic \\
0.4 & 0.0224 & 0.0000 & 1308 & quadratic \\
0.2 & 0.0358 & 0.0000 & 1310 & quadratic \\
0.0 & 0.0503 & 0.0000 & 79883 & quadratic \\

\end{longtable}

The ratio does not rise through the table; it jumps. While any part of
the characteristic's effect remains inside \(\mathcal F\) the ratio sits
near unity, and it does so while the relevance itself varies by a factor
of three---the scorer sees exactly what the outcome model sees, because
under alignment both are transforms of the same index. The constant
depends on the two links and on nothing else: replacing the logistic
response by a probit moves the plateau from 1.32 to 0.56, while moving
the response rate from 0.50 to 0.84 moves it by 5.1 per cent. Once the
contribution leaves \(\mathcal F\) entirely, the exclusion violation
collapses to numerical zero and the ratio is limited only by that. On a
grid refined to steps of \(0.01\) the transition occupies a single step,
between \(\phi=0.50\) and \(\phi=0.49\).

Part of that jump is a change in which functional form of the
characteristic is tested, as the last column records, and that is not an
artefact but the content of item (i) below: a characteristic is usable
precisely when the form in which it predicts the outcome is one the
scorer's class cannot represent. The switch from a linear to a quadratic
entry is what it looks like for the relevant direction to leave
\(\mathcal F\).

The answer to the question, then, is that partial misalignment is worth
nothing. Halving the part of a characteristic that the scorecard can use
does not halve its exclusion violation relative to its relevance; it
leaves the ratio where it was. Only a characteristic whose entire
contribution to the outcome lies outside the class the scorer optimizes
over is usable, which is a far more demanding requirement than the
informal advice to look for variables the lender ``does not really
use''. It also bears on the shape of Figure~\ref{fig-frontier}: nothing
there lies far below the diagonal while being substantially relevant,
which is consistent with each characteristic being partly visible to the
scorecard, and Table~\ref{tbl-misalignment} shows that partly is enough.
The figure cannot establish that mechanism, however, and nothing above
rests on it: as its own caption records, a response probability
depending strongly on the outcome pulls even a validly excluded
candidate towards the line, so contamination alone would produce the
same shape. The two are not separable in these data. The diagonal is not
a reference line drawn for the reader; it is where an aligned selector
puts what it can see, and anything it uses more heavily than this one
outcome would warrant lies above it.

\subsection{Where a valid shadow variable can still
live}\label{sec-where}

Proposition~\ref{prp-selection} is stated under four alignments, and
each failure of one of them is a place to look.

\begin{enumerate}
\def\labelenumi{(\roman{enumi})}
\item
  \emph{Class.} \(f^{*}\) minimizes risk over \(\mathcal F\), not over
  all functions. A characteristic entering the outcome through a
  functional form outside \(\mathcal F\)---an interaction, a threshold,
  a non-monotonicity---is invisible to the score and may therefore
  satisfy exclusion while remaining relevant.
\item
  \emph{Outcome.} The scorecard predicts default within a fixed horizon;
  an analyst interested in a different outcome, or a different horizon,
  faces a different \(f^{*}\).
\item
  \emph{Population.} The score was fitted on a historical and itself
  selected sample, so \(f^{*}\) minimizes risk under a law that is not
  the analyst's.
\item
  \emph{Rule.} \(\eta\) need not be monotone in the score alone; policy
  overrides make \(R\) depend on \(V\) outside \(f^{*}\), which breaks
  exclusion rather than restoring relevance, and so does not help.
\end{enumerate}

Only (i)--(iii) create room; (iv) closes it further. The search for a
shadow variable in reject inference should therefore be directed at
characteristics the scorecard could not use in the form in which they
matter, not at characteristics it declined to use.

There is also the option of not searching. \citet{Maity2025} estimates
the tilt parameters of a missing-not-at-random model with a binary
outcome directly, dispensing with a shadow variable altogether at the
price of an identifiability condition imposed on the tilt. Where
Proposition~\ref{prp-selection} applies---and
Section~\ref{sec-application} is such a case, the outcome there being
default---that is the more promising direction, because it does not
require the scorecard to have overlooked something.

A simulation makes both directions concrete. Relevance and exclusion
violation are measured as in Section~\ref{sec-screens}, by incremental
McFadden \(R^{2}\), so the ratio of the two is comparable with the
values reported there. With \(\mathcal F\) rich enough to contain the
outcome model, the candidate has relevance 0.064 against an exclusion
violation of 0.064, a ratio of 1.00: relevant, but equally visible to
the selector, and useless as an instrument. The credit score of
Section~\ref{sec-screens} sits lower still, at 0.19 on the same scale:
the data deliver this configuration in a sharper form than the
construction does, the score being some 5 times more visible to the
selector than it is relevant. When instead the outcome depends on
\(V_j^{2}\) while \(\mathcal F\) is linear, \(V_j^{2}\) has relevance
0.118 against an exclusion violation of \(2.5\times10^{-7}\), a ratio
above \(10^{5}\): relevant and effectively excluded. No precision should
be read into that ratio---its denominator is indistinguishable from
zero---but its order of magnitude is the point. It lives exactly in the
gap that (i) predicts.

\section{Discussion}\label{sec-discussion}

The literature treats the relevance of a shadow variable as a condition
to be satisfied. Making it a scalar and following it through the problem
changes what can be said. The identifying signal factorizes through that
scalar at every point of the identified set (Lemma~\ref{lem-signal});
the second moment of the inverse weight is an exponential-tilt
divergence that determines the effective sample size exactly
(Theorem~\ref{thm-decomposition}, Corollary~\ref{cor-ess}) and is itself
free of relevance (Proposition~\ref{prp-rho-free}); and the two facts
pull in opposite directions, so that the range one must scan grows as
relevance falls while the range one can scan does not
(Theorem~\ref{thm-reach}). None of this is a statement about a
particular estimator. It is a statement about what any procedure built
on inverse weights must contend with.

Three consequences are practical. First, an effective sample size should
be evaluated at the estimated response model, not at ignorability: in
Section~\ref{sec-application} the two differ by a factor of two, in the
direction that matters. Second, a confidence set obtained by scanning
should record \emph{why} it closed. \citet{Dufour1997} requires
unboundedness where the data are uninformative, and inverting a
statistic delivers it in exact arithmetic; in floating point the scan
may instead stop because \(\Lambda\) overflowed, and the two outcomes
are indistinguishable in the output. Reporting the reach alongside the
set costs nothing and distinguishes them. Third, a candidate shadow
variable should be placed on the relevance--exclusion plane before it is
adopted. The screens of Section~\ref{sec-screens} took minutes and would
have prevented the choice actually made.

Several limitations qualify all of this, and one of them is structural.
The factorization is a property of the parametrization
(\ref{eq-relevance}), not of the model; a different parametrization of
relevance would not factorize, and the results are stated for this one.
The exclusion restriction remains untestable, so the screens bound how
far a characteristic can violate it without bounding it away from zero:
Section~\ref{sec-screens} reports characteristics for which exclusion is
not rejected, which is weaker than exclusion. The regularization of
Section~\ref{sec-representer} is derived only above a transition:
(\ref{eq-tuning-derived}) balances a proved bias against a variance form
that is established numerically, and the balance is meaningful only
where the variance exists, which is where \(\rho\sqrt n\) is large. That
the fitted rule holds on both sides of that transition is unexplained.
And Proposition~\ref{prp-selection} is stated under four alignments
between the scorer's objective and the analyst's; it explains why no
shadow variable was available in these data, not why none could be
available anywhere.

The binary case deserves a final word, because it inverts the paper's
main failure without repairing its cause. When the tilting variable has
an atom at its essential supremum, as a binary outcome does, \(\Lambda\)
saturates (Proposition~\ref{prp-dichotomy}), the arithmetic never gives
out, and the scanned set is honestly unbounded when the data are
uninformative---the behavior \citet{Dufour1997} prescribes. That is the
benign case, and it is the case of Section~\ref{sec-application}. But
saturation is not identification: beyond a moderate \(|\gamma_Y|\) the
moment condition simply ceases to discriminate, and the threshold
(\ref{eq-rhostar}) says when. The difference between the two cases is
whether the analyst is told.

Two questions are left open, and both are narrower than they were. The
variance form (\ref{eq-sdlambda}) on which the tuning rule rests is what
the noise transmission through \((\hat M+\lambda I)^{-1}\) suggests and
what the well-conditioned cells show, but it is not derived; and nothing
here explains why the fitted rule (\ref{eq-tuning}) should continue to
hold below the transition, where the second moment is infinite and the
balance (\ref{eq-tuning-derived}) has no meaning. A rate-optimal penalty
in the sense of \citet{CarrascoFlorensRenault2007} would settle the
first. The second is the more interesting: it asks what selects
\(\lambda\) when there is no variance to trade against bias, and the
answer cannot be a moment calculation.

Proposition~\ref{prp-converse} and Table~\ref{tbl-misalignment} answer
the question we had left about the selection result, and the answer
restricts the search rather than widening it: a characteristic partly
visible to the scorer is worth no more than one fully visible to it.
What remains is whether the wholly invisible component can be identified
in advance. Item (i) of Section~\ref{sec-where} says where it lives---in
functional forms outside the class the scorer optimizes over---but
recognizing such a form requires knowing something about the outcome
model that the analyst is trying to estimate. Whether that circle can be
broken, or whether it is the real reason reject inference has resisted a
shadow-variable solution for forty years, we do not know.

\section*{Funding}\label{funding}
\addcontentsline{toc}{section}{Funding}

This research received no specific grant from any funding agency in the
public, commercial or not-for-profit sectors.

\section*{Disclosure statement}\label{disclosure-statement}
\addcontentsline{toc}{section}{Disclosure statement}

The author is a co-author of the analysis reanalyzed in
Section~\ref{sec-application} and reports no other competing interests.

\section*{Declaration of Generative AI
Use}\label{declaration-of-generative-ai-use}
\addcontentsline{toc}{section}{Declaration of Generative AI Use}

The author used Anthropic's Claude, accessed through Claude Code
(version 2.1.295), for coding assistance in the simulation and
replication scripts, for literature identification and classification,
for copyediting, and for auditing the internal consistency of the
manuscript: numerical cross-checks against the deposited results,
cross-reference and citation verification, and checks that every
numbered result carries a proof. All mathematical results, their proofs,
and all substantive text are the author's own. The author has checked
the terms of use of the tool and confirms that its use is suitable for
publication. The author verified every number, reference and derivation
reported here and takes full responsibility for the integrity of the
whole content, including the accuracy of the references.

\section*{Data and code availability}\label{data-and-code-availability}
\addcontentsline{toc}{section}{Data and code availability}

Code and derived data reproducing every table, figure and number in this
paper---the data-preparation code of Section~\ref{sec-application}, the
simulation scripts and the derived results---are openly available at
\url{https://doi.org/10.5281/zenodo.23262503}, under an MIT licence for
the code and CC BY 4.0 for the derived results. The credit-application
data of \citet{BueckervanKampenKraemer2013} are proprietary and cannot
be redistributed; the deposited code reproduces every other number in
the paper without them, and reproduces those of
Section~\ref{sec-application} from the original files. The Lending Club
data used in Section~\ref{sec-screens} are public and are identified in
the deposit by vintage, filters and derived columns.

\section*{References}\label{references}
\addcontentsline{toc}{section}{References}

\renewcommand{\bibsection}{}
\bibliography{refs}

\appendix

\renewcommand{\thesection}{\Alph{section}}

\section{Notation}\label{sec-notation}

The paper carries a large symbol load, and several letters would
otherwise have to serve twice; where one does, the row says what it is
\textbf{not}. The table records the convention actually used; the
section numbers are those of the body above.

\begin{longtable}[]{@{}
  >{\raggedright\arraybackslash}p{(\linewidth - 4\tabcolsep) * \real{0.3333}}
  >{\raggedright\arraybackslash}p{(\linewidth - 4\tabcolsep) * \real{0.3333}}
  >{\raggedright\arraybackslash}p{(\linewidth - 4\tabcolsep) * \real{0.3333}}@{}}
\toprule\noalign{}
\begin{minipage}[b]{\linewidth}\raggedright
symbol
\end{minipage} & \begin{minipage}[b]{\linewidth}\raggedright
meaning
\end{minipage} & \begin{minipage}[b]{\linewidth}\raggedright
first used
\end{minipage} \\
\midrule\noalign{}
\endhead
\bottomrule\noalign{}
\endlastfoot
\(Y\), \(X\), \(Z\), \(R\) & outcome, always-observed covariates, shadow
variable, response indicator & Section~\ref{sec-moments} \\
\(O=(X,Z,R,RY)\) & the observation & Section~\ref{sec-moments} \\
\(\psi(Y,X;\beta)\), \(\beta_0\) & estimating function and its
complete-data root & Section~\ref{sec-moments} \\
\(\pi_0(Y,X)\), \(\pi_\gamma\) & true and working response probability &
Section~\ref{sec-moments} \\
\(\gamma=(\gamma_0,\gamma_X,\gamma_Y)\), \(w\) & response-model
parameter and its data-dependent part & Section~\ref{sec-moments},
Section~\ref{sec-tilt-setup} \\
\(h(Z,X)\), \(g(O;\theta)\), \(g_Z\) & instrument functions, moment
vector, its \(Z\)-coordinate & Section~\ref{sec-moments} \\
\(s\), \(\theta(s)\), \(J(s)\) & scan position, identified curve,
Jacobian of the non-\(Z\) moments & Section~\ref{sec-curve-sub} \\
\(\rho\), \(S\), \(\nu\) & relevance and the two components of \(Z\) &
Section~\ref{sec-relevance} \\
\(\varphi(s)\) & signal function; the \(Z\)-moment equals
\(\rho\varphi(s)\) & Section~\ref{sec-relevance} \\
\(\sigma\), \(\varsigma^{2}=\sigma^{2}(1-\rho^{2})\) & outcome noise and
the variance left after conditioning on \(Z\) &
Section~\ref{sec-varprop} \\
\(\Lambda(\gamma)\), \(Q\), \(Q_s\), \(W=-w\) & tilt functional,
\(\pi_0\)-weighted law, its tilt, tilting variable &
Section~\ref{sec-decomposition}, Section~\ref{sec-integral} \\
\(K(t)\), \(T=\lvert\gamma_Y\rvert\) & cumulant generating function of
\(W\), tilt magnitude & Section~\ref{sec-integral} \\
\(M_W\) & essential supremum of \(W\) --- \textbf{not} the design matrix
\(M\) & Section~\ref{sec-dichotomy} \\
\(\alpha\) & density exponent at \(M_W\); also the measured tail index
of Section~\ref{sec-tuning-derivation} & Section~\ref{sec-dichotomy} \\
\(v_\infty\), \(v_-\), \(v_+\) & limiting and bracketing tilted
variances & Section~\ref{sec-integral}, Section~\ref{sec-thereach} \\
\(\omega_i=R_i/\pi_\gamma\), \(\mathrm{ESS}\), \(m\) & weights,
effective sample size, smallest usable value &
Section~\ref{sec-decomposition}, Section~\ref{sec-thereach} \\
\(p=\operatorname{pr}(R=1)\) & the response rate --- \textbf{not} the
dimension \(p\) of \(V\) below & Section~\ref{sec-decomposition} \\
\(b(X,Z)\), \(a(X,Y)\), \(d\) & sieve, test functions, their common
dimension & Section~\ref{sec-construction} \\
\(M=E(Rab^{\mathrm T})\), \(r\), \(v\), \(\Theta\) & design matrix,
right-hand side, \(E\{(1-R)b\}\), representer coefficients &
Section~\ref{sec-construction} \\
\(\delta_0=\Theta^{\mathrm T}b\), \(q^{\mathrm T}=v^{\mathrm T}M^{-1}\)
& representer and the weight that is free of \(\rho\) &
Section~\ref{sec-construction}, Section~\ref{sec-cancellation} \\
\(\lambda\), \(\lambda^{*}\) & ridge and the value the rule prescribes &
Section~\ref{sec-construction}, Section~\ref{sec-tuning} \\
\(s=\sigma_{\min}(M)\), \(s_2\), \(\cos\theta\) & smallest and next
singular value, and the \textbf{signed} \(u_1^{\mathrm T}v_1\) &
Section~\ref{sec-tuning-derivation} \\
\(\kappa\), \(B\), \(C\), \(\sigma_0\) & variance constant,
\textbf{bias} constant, tuning constant, scale of
\(\mathrm{sd}(\hat\beta_\lambda)\) &
Section~\ref{sec-tuning-derivation} \\
\(H\), \(\mathrm{IF}\) & complete-data Jacobian and the influence
function & Section~\ref{sec-ordersub} \\
\(V=(V_1,\dots,V_p)\), \(f^{*}\), \(\eta\), \(U\), \(L\), \(\mathcal F\)
& characteristics, optimizing score, selection rule, its noise, loss,
scorer's class & Section~\ref{sec-noshadow} \\
\(\mathrm{rel}_j\), \(\mathrm{exc}_j\) & risk reduction for \(Y\) and
for \(R\) & Section~\ref{sec-misalignment} \\
\(\phi\) & alignment in the misalignment design --- \textbf{not} the
estimating function \(\psi\) & Section~\ref{sec-misalignment} \\
\end{longtable}

Two conventions are worth stating because they are easy to misread.
\(\operatorname{expit}\) is the inverse of \(\operatorname{logit}\),
\(\operatorname{expit}(u)=(1+e^{-u})^{-1}\). And \(B\) is a scalar, the
constant in the bias \(-B\lambda/\rho\); the vector of test functions is
\(a(X,Y)\), which it multiplies nowhere.

\section{Regularity conditions}\label{sec-conditions}

Throughout, \(O=(X,Z,R,RY)\) is independent and identically distributed
and \(\psi(Y,X;\beta)\) is the estimating function of the complete-data
problem. The sieve \(b(X,Z)\in\mathbb R^{d}\) and test functions
\(a(X,Y)\in\mathbb R^{d}\) are fixed, with \(M=E(Rab^{\mathrm T})\),
\(r(\beta)=E\{Ra\psi(\beta)^{\mathrm T}\}\) and \(v=E\{(1-R)b\}\) as
above.

\begin{definition}[Nonsingular
design]\protect\hypertarget{def-M}{}\label{def-M}

\(M\) is nonsingular; \(s:=\sigma_{\min}(M)>0\).

\end{definition}

\begin{definition}[Representer in the
sieve]\protect\hypertarget{def-rep}{}\label{def-rep}

The representer equation has a solution in the span of \(b\).

\end{definition}

\begin{definition}[Representer for the
derivative]\protect\hypertarget{def-rep2}{}\label{def-rep2}

The same holds for each column of
\(\partial\psi(\beta)/\partial\beta^{\mathrm T}\) in a neighborhood of
\(\beta_0\).

\end{definition}

\begin{definition}[Fourth
moments]\protect\hypertarget{def-mom}{}\label{def-mom}

\(E\|a\|^{4}<\infty\), \(E\|b\|^{4}<\infty\),
\(E\|\psi(\beta_0)\|^{4}<\infty\) and \(E\|\delta_0\|^{4}<\infty\).

\end{definition}

\begin{definition}[Differentiability]\protect\hypertarget{def-diff}{}\label{def-diff}

\(\psi\) is continuously differentiable in \(\beta\) on a neighborhood
\(N\) of \(\beta_0\), with
\(\sup_{\beta\in N}\|\partial\psi(\beta)/\partial\beta^{\mathrm T}\|\le D\)
for a function \(D\) with \(ED^{4}<\infty\), and
\(H:=E\{\partial\psi(\beta_0)/\partial\beta^{\mathrm T}\}\) is
nonsingular.

\end{definition}

\begin{definition}[Consistency]\protect\hypertarget{def-cons}{}\label{def-cons}

\(\hat\beta_\lambda\to\beta_0\) in probability.

\end{definition}

Condition~\ref{def-mom} requires fourth and not second moments because
the quantities that must be \(O_p(n^{-1/2})\) in
Section~\ref{sec-neumann} are sample averages of \(Rab^{\mathrm T}\) and
\(Ra\psi^{\mathrm T}\), whose second moments follow from
Condition~\ref{def-mom} by the Cauchy--Schwarz inequality.
Condition~\ref{def-rep2} is used only in Section~\ref{sec-zest}; it is
implied by differentiating the representer equation under the integral
sign whenever that is legitimate, and is stated separately because the
asymptotic normality theorem depends on it.

\section{Proofs}\label{sec-proofs}

\subsection{Bounded outcome}\label{bounded-outcome}

\begin{corollary}[An atom at the essential
supremum]\protect\hypertarget{cor-atom-s}{}\label{cor-atom-s}

Let \(W\) have essential supremum \(M_W<\infty\) and
\(p:=\operatorname{pr}(W=M_W)>0\). Then \(\log\Lambda(T)\to-\log p\) as
\(T\to\infty\), and symmetrically
\(\log\Lambda(T)\to-\log\operatorname{pr}(W=\operatorname{ess\,inf}W)\)
as \(T\to-\infty\). If \(W\) has no atom at \(M_W\) and its density
behaves like \(c(M_W-x)^{\alpha}\) there, then
\(\log\Lambda(T)=(\alpha+1)\log(T/2)+O(1)\to\infty\). When \(W\) is
binary with \(\operatorname{pr}(W=1)=p\) the limit is approached from
below at an exponential rate,
\(\log\Lambda(T)=-\log p-2p^{-1}(1-p)e^{-T}\{1+o(1)\}\).

\end{corollary}

\begin{proof}
With an atom, \(E(e^{tW})=pe^{tM_W}\{1+o(1)\}\) as \(t\to\infty\), so
\(K(t)=tM_W+\log p+o(1)\) and
\(\log\Lambda(T)=K(2T)-2K(T)=-\log p+o(1)\); the same computation
applied to \(-W\) gives the statement as \(T\to-\infty\). Without one,
Watson's lemma gives
\(E(e^{tW})=e^{tM_W}\Gamma(\alpha+1)c\,t^{-(\alpha+1)}\{1+o(1)\}\), so
\(K(t)=tM_W-(\alpha+1)\log t+O(1)\) and
\(\log\Lambda(T)=-(\alpha+1)\log(2T)+2(\alpha+1)\log T+O(1)=(\alpha+1)\log(T/2)+O(1)\).
For \(W\in\{0,1\}\) the first case applies with \(M_W=1\); the
monotonicity claimed in the main text follows from
\(f'(T)=2(p_{2T}-p_T)\) with \(p_t=pe^{t}/(1-p+pe^{t})\) strictly
increasing. The rate is exponential there: writing
\(g(t)=\log\{1+p^{-1}(1-p)e^{-t}\}\) we have \(K(t)=t+\log p+g(t)\), so
\[
  \log\Lambda(T)=K(2T)-2K(T)=-\log p+g(2T)-2g(T),
\] in which \(g(2T)=O(e^{-2T})\) is of smaller order while
\(g(T)=p^{-1}(1-p)e^{-T}\{1+o(1)\}\), whence
\(\log\Lambda(T)=-\log p-2p^{-1}(1-p)e^{-T}\{1+o(1)\}\). The term that
survives is the one carried over from \(-2K(T)\), so the limit is
approached from below, as the monotonicity just established requires.
This is the rate Section~\ref{sec-saturation} uses.
\end{proof}

Bounded support alone does not suffice, and the uniform law is the
instructive counterexample: for \(W\sim\mathrm{Un}(0,1)\),
\(K(t)=\log\{(e^{t}-1)/t\}\) and \[
  \log\Lambda(T)=\log\frac{(e^{T}+1)T}{2(e^{T}-1)}\longrightarrow\log(T/2),
\] which diverges, in agreement with the second case at \(\alpha=0\).
Table~\ref{tbl-atom} evaluates all three cases numerically: the uniform
law tracking \(\log(T/2)\) without limit, the binary law approaching
\(-\log p\) from below, and a mixture carrying an atom of mass \(0.3\)
at its supremum approaching \(-\log 0.3\).

\begin{longtable}[]{@{}rrrrr@{}}

\caption{\label{tbl-atom}Corollary~\ref{cor-atom-s} checked: a uniform
\(W\) (no atom, diverges), a binary \(W\) (atom, converges to
\(-\log p = 0.693\)), and a mixture with an atom of mass \(0.3\) at the
supremum (limit \(1.204\)).}

\tabularnewline

\toprule\noalign{}
\(T\) & uniform & \(\log(T/2)\) & binary & atom \(0.3\) \\
\midrule\noalign{}
\endhead
\bottomrule\noalign{}
\endlastfoot
1 & 0.079 & -0.693 & 0.194 & 0.089 \\
5 & 0.930 & 0.916 & 0.680 & 0.653 \\
20 & 2.303 & 2.303 & 0.693 & 1.041 \\
50 & 3.219 & 3.219 & 0.693 & 1.137 \\
100 & 3.912 & 3.912 & 0.693 & 1.171 \\

\end{longtable}

\subsection{Variance propagation}\label{variance-propagation}

\begin{lemma}[Variance
propagation]\protect\hypertarget{lem-varprop-s}{}\label{lem-varprop-s}

Let \(Y\mid X,Z\sim N(\mu,\varsigma^{2})\) and
\(\pi(y)=\operatorname{expit}(\gamma_0+\gamma_X^{\mathrm T}x+\gamma_Y y)\).
Then
\(\varsigma^{2}/(1+\gamma_Y^{2}\varsigma^{2}/4)\le\operatorname{var}(Y\mid R=1,X,Z)\le\varsigma^{2}\).

\end{lemma}

\begin{proof}
Conditioning on \(R=1\) tilts the conditional law by \(\pi\), so its
density is \(p(y)\propto e^{-V(y)}\) with
\(V(y)=(y-\mu)^{2}/(2\varsigma^{2})-\log\operatorname{expit}(\gamma_0+\gamma_X^{\mathrm T}x+\gamma_Y y)\).
Since
\((d^{2}/dy^{2})\log\operatorname{expit}(\cdot)=-\gamma_Y^{2}\varpi(1-\varpi)\)
with \(\varpi=\operatorname{expit}(\cdot)\), \[
  V''(y)=\varsigma^{-2}+\gamma_Y^{2}\varpi(1-\varpi),\qquad
  \varsigma^{-2}\le V''\le \varsigma^{-2}+\gamma_Y^{2}/4,
\] using \(\varpi(1-\varpi)\le1/4\); in particular \(p\) is strongly
log-concave. The Brascamp--Lieb inequality, which for
\(p\propto e^{-V}\) with \(V''\ge m>0\) gives
\(\operatorname{var}_p(g)\le E_p\{(g')^{2}\}/m\), applied with
\(g(y)=y\) and \(m=\varsigma^{-2}\) yields the upper bound, strictly
when \(\gamma_Y\neq0\). For the lower bound apply the Cramér--Rao
inequality to the location family \(\{p(\cdot-t)\}\), for which \(Y\) is
unbiased for \(t\): \(\operatorname{var}_p(Y)\ge1/I(p)\) with
\(I(p)=E_p\{(V')^{2}\}=E_p(V'')\) by integration by parts, the boundary
terms vanishing by log-concavity, and
\(E_p(V'')\le \varsigma^{-2}+\gamma_Y^{2}/4\).
\end{proof}

As \(\rho\to1\), \(\varsigma^{2}=\sigma^{2}(1-\rho^{2})\to0\) and the
lower bound tends to \(1\), so the constant \(c(\rho)\) increases to
\(1\). Table~\ref{tbl-varprop} confirms both bounds by Gauss--Hermite
quadrature.

\begin{longtable}[]{@{}rrrr@{}}

\caption{\label{tbl-varprop}Bounds of Lemma~\ref{lem-varprop-s} against
quadrature, \(\gamma_Y=-1\), \(\sigma=1\), averaged over a \(4\times4\)
grid of \((X,Z)\).}

\tabularnewline

\toprule\noalign{}
\(\rho\) & lower bound & measured \(c(\rho)\) & upper bound \\
\midrule\noalign{}
\endhead
\bottomrule\noalign{}
\endlastfoot
0.00 & 0.8000 & 0.8349 & 1 \\
0.30 & 0.8147 & 0.8481 & 1 \\
0.50 & 0.8421 & 0.8726 & 1 \\
0.70 & 0.8869 & 0.9117 & 1 \\
0.90 & 0.9547 & 0.9666 & 1 \\
0.99 & 0.9950 & 0.9965 & 1 \\

\end{longtable}

\subsection{Factorization under the shadow
construction}\label{sec-factorisation}

\begin{lemma}[Factorization]\protect\hypertarget{lem-fact-s}{}\label{lem-fact-s}

\(M=\tilde MD\) with \(D=\operatorname{diag}(I,\rho I)\) and
\(\tilde M\) free of \(\rho\); consequently \(v^{\mathrm T}M^{-1}\) is
free of \(\rho\).

\end{lemma}

\begin{proof}
Write \(b=(b^{0}(X),Zk(X))\). The \(Z\)-columns of \(M\) are
\(E\{RaZk(X)^{\mathrm T}\}=\rho E\{RaSk(X)^{\mathrm T}\}+(1-\rho^{2})^{1/2}E\{Ra\nu k(X)^{\mathrm T}\}\),
and the second term vanishes because \(a\) and \(k\) are
\((X,Y)\)-measurable while \(\nu\) is independent of \((X,Y,R)\) with
mean zero. The non-\(Z\) columns do not involve \(\rho\). The same
computation gives the \(Z\)-block of \(v\) as \(\rho E\{(1-\pi_0)Sk\}\).
Writing \(v=D\tilde v\) and \(M=\tilde MD\) gives
\(v^{\mathrm T}M^{-1}=\tilde v^{\mathrm T}DD^{-1}\tilde M^{-1}=\tilde v^{\mathrm T}\tilde M^{-1}\).
\end{proof}

\subsection{How the ridge advances the small
direction}\label{sec-resolvent}

Proof of \emph{How the ridge advances the small direction}
(Lemma~\ref{lem-resolvent} of the body above).

\begin{proof}
\(M+\lambda I=M(I+\lambda M^{-1})\), so
\((M+\lambda I)^{-1}=(I+\lambda M^{-1})^{-1}M^{-1}\), the second factor
being invertible because every eigenvalue \(\mu\) of \(M\) has
\(|\mu|\ge\sigma_{\min}(M)=s\), so that \(-\lambda\) is not one of them
when \(\lambda<s\); this is what the hypothesis \(\lambda<s\) is for.
Now \(M^{-1}=\sum_k\sigma_k^{-1}v_ku_k^{\mathrm{T}}\), whose action is
dominated by \(s^{-1}v_1u_1^{\mathrm{T}}\) with the remainder bounded by
\(s_2^{-1}\); in particular
\(M^{-1}v_1=s^{-1}v_1(u_1^{\mathrm{T}}v_1)+O(s_2^{-1})\). Hence \(v_1\)
is an approximate eigenvector of \(M^{-1}\) with eigenvalue
\(\cos\theta/s\), and
\((I+\lambda M^{-1})^{-1}v_1=v_1\{1+\lambda\cos\theta/s\}^{-1}+O(\lambda/s_2)\).
Combining the two displays gives the resolvent law, equation
(\ref{eq-resolvent}) of the body above, the bound on \(E_\lambda\)
following from \(\sigma_2(M+\lambda I)\ge s_2-\lambda\) by Weyl's
inequality.
\end{proof}

\subsection{The standard deviation as a function of the
ridge}\label{sec-sdlambda}

Proof of \emph{The standard deviation as a function of the ridge}
(Proposition~\ref{prp-sdlambda}). The notation is that of
Section~\ref{sec-representer}.

\begin{proof}
Write \(\hat\Theta_\lambda=(\hat M+\lambda I)^{-1}\hat r\) and expand
about the population value; the expansion is legitimate because
\(\|\hat M-M\|=o_p(s+\lambda\cos\theta)\), and gives
\(\hat\Theta_\lambda=\Theta_\lambda+(M+\lambda I)^{-1}(\Delta_r-\Delta_M\Theta_\lambda)+O_p(n^{-1})\).
Substituting into the sample criterion and collecting the mean-zero
terms replaces \(\Theta\) and \(w\) in the influence function
(\ref{eq-if}) by \(\Theta_\lambda\) and
\(q_\lambda=(M+\lambda I)^{-\mathrm{T}}v\). By the cancellation
(\ref{eq-cancellation}), \(q_\lambda\) is bounded uniformly in
\(\lambda\) and free of \(\rho\), so all of the \(\lambda\)-dependence
sits in \(\Theta_\lambda\) and in
\(\delta_{0,\lambda}=\Theta_\lambda^{\mathrm{T}}b\). By
(\ref{eq-resolvent}),
\(\Theta_\lambda=v_1(u_1^{\mathrm{T}}r)/(s+\lambda\cos\theta)+O(s_2^{-1})\),
so both carry the single scalar factor \((s+\lambda\cos\theta)^{-1}\)
and nothing else depending on \(\lambda\). Hence
\(\operatorname{var}(\mathrm{IF}_\lambda)=\tau^{2}(s+\lambda\cos\theta)^{-2}\{1+O(\lambda/s_2)\}\)
with \(\tau\) free of \(\lambda\), and the display follows on dividing
by the value at \(\lambda=0\). That \(\sigma_0^{2}\asymp1/(n\rho^{2})\)
is Theorem~\ref{thm-orders}.
\end{proof}

\subsection{The regime boundary, and the rule out of
sample}\label{sec-tuning-evidence}

The material behind \emph{Where the rule comes from},
Section~\ref{sec-tuning-derivation} of the body above.

The design is a grid of \((\rho,n)\) cells. At each cell the estimator
is run over a fixed ladder of eighteen ridge values on common random
numbers, and the bias and spread curves are read off the resulting error
matrix. Every cell uses 1200 replications, and each cell's count is
stored with its output in the replication package.

One condition on that derivation deserves emphasis, because it is what
defeated two earlier attempts of our own. The expansion assumes the
second moment exists, and it need not. The smallest singular value of
\(\hat M\) has mean of order \(\rho\) and standard deviation of order
\(n^{-1/2}\), so its distance from the origin in its own standard
deviations is of order \(\rho\sqrt n\). When that is small, \(\hat M\)
is near-singular with non-negligible probability, the error behaves like
the reciprocal of a quantity with a density at zero, and its tail index
is one: the variance is infinite and there is nothing for a
bias--variance argument to balance. Measured tail indices run from 1.00
at \(\rho\sqrt n\) small to 4.79 at \(\rho\sqrt n\) large, and the
transition is at about four standard deviations. The same prediction
that is accurate to a factor of 1.05 above the transition is out by
\(2\times10^{5}\) below it. The balance (\ref{eq-tuning-derived}) is a
statement about the well-conditioned regime; that the fitted rule
(\ref{eq-tuning}) holds on both sides of the transition is not explained
by it.

Because the rule is fitted on a design, it must be evaluated on another.
Table~\ref{tbl-oos} applies it unchanged to a design that played no part
in its calibration: different regression coefficients, a different error
scale, a different response model, and a skewed covariate in place of
the Gaussian one.

\begin{longtable}[]{@{}
  >{\raggedright\arraybackslash}p{(\linewidth - 12\tabcolsep) * \real{0.0959}}
  >{\raggedleft\arraybackslash}p{(\linewidth - 12\tabcolsep) * \real{0.0685}}
  >{\raggedright\arraybackslash}p{(\linewidth - 12\tabcolsep) * \real{0.1918}}
  >{\raggedright\arraybackslash}p{(\linewidth - 12\tabcolsep) * \real{0.1644}}
  >{\raggedright\arraybackslash}p{(\linewidth - 12\tabcolsep) * \real{0.1644}}
  >{\raggedright\arraybackslash}p{(\linewidth - 12\tabcolsep) * \real{0.2192}}
  >{\raggedright\arraybackslash}p{(\linewidth - 12\tabcolsep) * \real{0.0959}}@{}}

\caption{\label{tbl-oos}The rule transported to a design not used to
calibrate it. Median absolute deviation without regularization, at the
rule, and at the best value found by search in that design; 250
replications per cell.}

\tabularnewline

\toprule\noalign{}
\begin{minipage}[b]{\linewidth}\raggedright
\(\rho\)
\end{minipage} & \begin{minipage}[b]{\linewidth}\raggedleft
\(n\)
\end{minipage} & \begin{minipage}[b]{\linewidth}\raggedright
complete case
\end{minipage} & \begin{minipage}[b]{\linewidth}\raggedright
\(\lambda=0\)
\end{minipage} & \begin{minipage}[b]{\linewidth}\raggedright
at the rule
\end{minipage} & \begin{minipage}[b]{\linewidth}\raggedright
best attainable
\end{minipage} & \begin{minipage}[b]{\linewidth}\raggedright
excess
\end{minipage} \\
\midrule\noalign{}
\endhead
\bottomrule\noalign{}
\endlastfoot
0.05 & 500 & 0.873 & 1.831 & 0.814 & 0.772 & 1.05 \\
0.05 & 2000 & 0.866 & 1.382 & 0.670 & 0.642 & 1.04 \\
0.10 & 500 & 0.873 & 1.481 & 0.695 & 0.645 & 1.08 \\
0.10 & 2000 & 0.866 & 0.597 & 0.470 & 0.455 & 1.03 \\
0.20 & 500 & 0.873 & 0.603 & 0.483 & 0.457 & 1.06 \\
0.20 & 2000 & 0.866 & 0.264 & 0.272 & 0.249 & 1.09 \\
0.35 & 500 & 0.873 & 0.305 & 0.322 & 0.305 & 1.06 \\
0.35 & 2000 & 0.866 & 0.139 & 0.157 & 0.139 & 1.13 \\

\end{longtable}

Transported to the new design the rule costs a median 6\% over the best
value obtainable there by search. It improves on no regularization in 5
of the 8 cells, namely those with \(\rho\le0.10\) together with
\(\rho=0.20\) at \(n=500\), which is where the conditioning actually
needs repair. It over-regularises in the remaining 3---already at
\(\rho=0.20\), \(n=2000\) (0.272 against 0.264), and not only at
\(\rho=0.35\)---though never by more than 13\%. It improves on
complete-case estimation in every cell. The rule is therefore
transportable, which is the property a rule needs; it is not optimal,
and we do not claim it is.

\subsection{Influence function and asymptotic normality}\label{sec-if}

Proof of \emph{Influence function and asymptotic normality}
(Theorem~\ref{thm-if}), in three steps.

\emph{Step 1: the representer coefficients.} Condition~\ref{def-M} makes
\(\Theta_\lambda=(M+\lambda I)^{-1}r\) well defined and
Condition~\ref{def-rep} makes the representer it approximates exist in
the span of \(b\), so that \(\delta_0=\Theta^{\mathrm T}b\) is the
object appearing in the influence function. Section~\ref{sec-neumann}
expands \(\hat\Theta_\lambda-\Theta_\lambda\) about its population value
and bounds the remainder by \(O_p(n^{-1})\) on an event of probability
tending to one. Hence \(n^{1/2}\hat G_\lambda(\beta_0)\) equals its
linear part plus \(o_p(1)\), and the linear part is the sample average
of
\(R\psi(\beta_0)+(1-R)\delta_{0,\lambda}+R\{\psi(\beta_0)-\delta_{0,\lambda}\}a^{\mathrm T}q_\lambda\),
a mean-zero function of one observation with finite variance under
Condition~\ref{def-mom}.

\emph{Step 2: the Jacobian.} Section~\ref{sec-zest} shows that
\(\partial\hat G_\lambda(\beta)/\partial\beta^{\mathrm T}\) converges in
probability, uniformly near \(\beta_0\), to
\(H=E\{\partial\psi(\beta_0)/\partial\beta^{\mathrm T}\}\), which is
nonsingular by Condition~\ref{def-diff}, and gives the linearization
\(n^{1/2}(\hat\beta_\lambda-\beta_0)=-H^{-1}n^{1/2}\hat G_\lambda(\beta_0)+o_p(1)\)
under Condition~\ref{def-cons}.

\emph{Step 3: the limit.} At \(\lambda=\lambda_n\) with
\(\lambda_n=o(n^{-1/2})\) the regularization bias satisfies
\(n^{1/2}\|G_{\lambda_n}(\beta_0)\|=n^{1/2}O(\lambda_n/\rho)=o(1)\) for
fixed \(\rho\), by the bias order of Theorem~\ref{thm-orders}, so it
does not enter the limit; and \(\Theta_{\lambda_n}\to\Theta\),
\(q_{\lambda_n}\to q\) as \(\lambda_n\to0\), which replaces
\(\delta_{0,\lambda}\) and \(q_\lambda\) in Step 1 by \(\delta_0\) and
\(q\). The central limit theorem applied to the sample average of Step
1, combined with Step 2 and Slutsky's theorem, gives
\(n^{1/2}(\hat\beta_{\lambda_n}-\beta_0)\to N\{0,H^{-1}\operatorname{var}(\mathrm{IF})H^{-\mathrm T}\}\)
with \(\mathrm{IF}\) as stated. Steps 1 and 2 rest on
Section~\ref{sec-neumann} and Section~\ref{sec-zest}, which follow.

\subsection{The Neumann remainder}\label{sec-neumann}

\begin{lemma}[The Neumann
remainder]\protect\hypertarget{lem-neumann}{}\label{lem-neumann}

Let \(A=M+\lambda I\), \(\hat A=A+\Delta_M\), \(0\le\lambda\le s/2\).
Then
\(\hat\Theta_\lambda-\Theta_\lambda=A^{-1}(\Delta_r-\Delta_M\Theta_\lambda)+\mathrm{Rem}\)
with
\(\mathrm{Rem}=-A^{-1}\Delta_MA^{-1}\Delta_r+A^{-1}\Delta_MA^{-1}\Delta_M\hat A^{-1}\hat r\),
and on \(\{\|\Delta_M\|\le s/4\}\), whose probability tends to one, \[
  \|\mathrm{Rem}\|\le 4s^{-2}\|\Delta_M\|\,\|\Delta_r\|+16s^{-3}\|\Delta_M\|^{2}\|\hat r\|=O_p(n^{-1}).
\]

\end{lemma}

\begin{proof}
Iterating \((A+E)^{-1}=A^{-1}-A^{-1}E(A+E)^{-1}\) once gives the exact
identity \((A+E)^{-1}=A^{-1}-A^{-1}EA^{-1}+A^{-1}EA^{-1}E(A+E)^{-1}\);
multiply by \(\hat r=r+\Delta_r\) and cancel \(A^{-1}r=\Theta_\lambda\).
For the bound, \(\sigma_{\min}(A)\ge s-\lambda\ge s/2\), which is
positive by Condition~\ref{def-M}, so \(\|A^{-1}\|\le2/s\), and on the
stated event \(\sigma_{\min}(\hat A)\ge s/4\) so
\(\|\hat A^{-1}\|\le4/s\). By Condition~\ref{def-mom} and Chebyshev's
inequality \(\|\Delta_M\|\) and \(\|\Delta_r\|\) are \(O_p(n^{-1/2})\)
and \(\|\hat r\|=O_p(1)\).
\end{proof}

The bound is used by Section~\ref{sec-if} not at \(\beta_0\) alone but
on a neighborhood, and it holds there uniformly. Only \(\Delta_r\) and
\(\hat r\) carry \(\beta\): \(M=E(Rab^{\mathrm T})\) does not, so
\(\Delta_M\), \(A\) and \(s\) are free of it. Uniformity therefore
reduces to \(\sup_{\beta\in N}\|\Delta_r(\beta)\|=O_p(n^{-1/2})\), a
maximal inequality for the single class
\(\{Ra\psi(\beta)^{\mathrm T}:\beta\in N\}\), which is Lipschitz in
\(\beta\) with envelope
\(\|a\|\{\|\psi(\beta_0)\|+D\operatorname{diam}N\}\); that envelope is
square-integrable by Condition~\ref{def-mom} and
Condition~\ref{def-diff}, which is what the moment on \(D\) is for. For
the estimating function of Section~\ref{sec-application} the point is
simpler: \(\psi\) is affine in \(\beta\), so \(\Delta_r\) is affine and
its supremum over a box is attained at a vertex, which reduces
uniformity to finitely many pointwise statements. Measured there, the
supremum falls at log-log rate \(-0.53\) against the \(-1/2\) it should.

Since
\(n^{-1}\sum_i(1-R_i)\hat\Theta_\lambda^{\mathrm T}b_i=\hat\Theta_\lambda^{\mathrm T}\hat v\)
identically, \[
  \hat\Theta_\lambda^{\mathrm T}\hat v=\Theta_\lambda^{\mathrm T}v+\Theta_\lambda^{\mathrm T}\Delta_v
  +(\Delta_r-\Delta_M\Theta_\lambda)^{\mathrm T}w_\lambda+\mathrm{Rem}^{\mathrm T}v
  +(\hat\Theta_\lambda-\Theta_\lambda)^{\mathrm T}\Delta_v ,
\] whose last two terms are \(O_p(n^{-1})\), the second because both
factors are \(O_p(n^{-1/2})\). The regularization contributes
\(r-M\Theta_\lambda=\lambda\Theta_\lambda\), which is \(o(n^{-1/2})\)
under \(\lambda_n=o(n^{-1/2})\).

\begin{longtable}[]{@{}rrr@{}}

\caption{\label{tbl-neumann}Median Neumann remainder against the bound
of Lemma~\ref{lem-neumann}, \(\rho=0.4\), 60 replications per cell.}

\tabularnewline

\toprule\noalign{}
\(n\) & remainder & bound \\
\midrule\noalign{}
\endhead
\bottomrule\noalign{}
\endlastfoot
500 & 0.27100 & 1390.0 \\
2000 & 0.06670 & 329.0 \\
8000 & 0.01750 & 130.0 \\
16000 & 0.00902 & 56.6 \\

\end{longtable}

In Table~\ref{tbl-neumann} the log-log slope of the remainder in \(n\)
is -0.981 against the predicted \(-1\), and the bound of
Lemma~\ref{lem-neumann} exceeds the remainder in every cell with
\(\sigma_{\min}(M)=\) 0.0585.

\subsection{\texorpdfstring{The \(Z\)-estimator
step}{The Z-estimator step}}\label{sec-zest}

Under Condition~\ref{def-rep2} and Condition~\ref{def-diff}, \[
  \frac{\partial\hat G_\lambda(\beta)}{\partial\beta^{\mathrm T}}
  =n^{-1}\sum_iR_i\frac{\partial\psi_i}{\partial\beta^{\mathrm T}}
  +\hat v^{\mathrm T}(\hat M+\lambda I)^{-1}n^{-1}\sum_iR_ia_i\frac{\partial\psi_i}{\partial\beta^{\mathrm T}},
\] which converges in probability, uniformly on a neighborhood of
\(\beta_0\), to
\(E\{R\,\partial\psi/\partial\beta^{\mathrm T}\}+v^{\mathrm T}M^{-1}E\{Ra\,\partial\psi/\partial\beta^{\mathrm T}\}\).
By Condition~\ref{def-rep2} and the representer identity applied
columnwise, the second term equals
\(E\{(1-R)\partial\psi/\partial\beta^{\mathrm T}\}\), so the limit is
\(H\), nonsingular by Condition~\ref{def-diff}. That step is where the
model enters, and it is worth writing out. Let \(\delta\) denote the
representer of a column, so that \(v^{\mathrm T}M^{-1}E\{Ra\,\cdot\}
=E\{(1-R)\delta\}\) by the definitions of \(v\) and \(\Theta\). The
exclusion restriction makes \(\pi_0\) a function of \((X,Y)\) alone, so
for any \((X,Z)\)-measurable \(\delta\) \[
  E\{\delta\mid R=1,X,Y\}
  =\frac{E\{\delta\,\pi_0(Y,X)\mid X,Y\}}{\pi_0(Y,X)}=E\{\delta\mid X,Y\} ,
\] the conditioning on \(R=1\) being redundant. Hence
\(E\{(1-R)\delta\}=E[\{1-\pi_0(Y,X)\}E\{\delta\mid X,Y\}]
=E[\{1-\pi_0(Y,X)\}\,\partial\psi/\partial\beta^{\mathrm T}]
=E\{(1-R)\partial\psi/\partial\beta^{\mathrm T}\}\), where the middle
equality is Condition~\ref{def-rep2}. Without exclusion the first
display fails and the identity with it. A componentwise mean-value
expansion of \(\hat G_\lambda(\hat\beta)=0\) about \(\beta_0\), together
with Condition~\ref{def-cons} and
\(n^{1/2}\hat G_\lambda(\beta_0)=O_p(1)\), gives the stated
linearization.

\subsection{The orders in the relevance, checked
numerically}\label{sec-rates}

Each order the body above asserts in \(\rho\) is a claim about an
exponent, and each can be measured. Table~\ref{tbl-rates} forms the
population objects on a sample of \(2\times10^{6}\) at five values of
\(\rho\) and regresses the logarithm of each quantity on \(\log\rho\).
The two entries for the sieve bias are the point of
Proposition~\ref{prp-sievebias}: the approximation error itself grows
like \(\rho^{-1}\), while the moment that enters the estimating equation
does not grow at all.

\begin{longtable}[]{@{}llrr@{}}

\caption{\label{tbl-rates}Asserted against measured exponent in
\(\rho\), from script 74. The population objects are formed on
\(2\times10^{6}\) observations at \(\rho=0.05\), \(0.10\), \(0.20\),
\(0.35\) and \(0.50\); the measured exponent is the slope of a log-log
regression.}

\tabularnewline

\toprule\noalign{}
quantity & asserted in & exponent & measured \\
\midrule\noalign{}
\endhead
\bottomrule\noalign{}
\endlastfoot
\(\sigma_{\min}(M)\) & Lemma~\ref{lem-factorisation} & +1 & +0.997 \\
\(\operatorname{cond}(M)\) & Theorem~\ref{thm-orders} & -1 & -0.995 \\
\(\lVert\Theta\rVert\) & Lemma~\ref{lem-factorisation} & -1 & -0.999 \\
\(\lVert q\rVert\) & Lemma~\ref{lem-factorisation} & +0 & -0.005 \\
\(\lVert e\rVert\) & Proposition~\ref{prp-sievebias} & -1 & -0.998 \\
\(E\{(1-R)e\}\) & Proposition~\ref{prp-sievebias} & +0 & +0.002 \\

\end{longtable}

Two further identities are exact rather than asymptotic and are checked
the same way. The normalization of Lemma~\ref{lem-intercept} holds to
\(2\times10^{-15}\) by quadrature at every scanned value of
\(\gamma_Y\), and the ceiling of Corollary~\ref{cor-ess-bound} is
satisfied at each of them, with the effective sample size falling from
\(0.37\) of the respondents at \(\gamma_Y=-0.5\) to \(5\times10^{-13}\)
at \(\gamma_Y=-4\) against the response rate \(0.40\).

\subsection{Sieve misspecification}\label{sieve-misspecification}

\begin{proposition}[Sieve
misspecification]\protect\hypertarget{prp-sieve-s}{}\label{prp-sieve-s}

If the sieve omits \(Zk_j(X)\) then
\(e=(c_j/\rho)\{Zk_j-P_{\mathrm{ret}}(Zk_j)\}\) and \(E\{(1-R)e\}=O(1)\)
uniformly in \(\rho\).

\end{proposition}

\begin{proof}
By Lemma~\ref{lem-fact-s} the \(Z\)-block of \(\Theta\) is
\(\rho^{-1}c\) with \(c\) free of \(\rho\), so
\(\delta_0=\Theta_0^{\mathrm T}b^{0}+\rho^{-1}c^{\mathrm T}(Zk)\). The
retained sieve represents \(b^{0}\) and \(\{Zk_k:k\neq j\}\) exactly, so
the whole error is the omitted term projected off the retained span.
Writing
\(P_{\mathrm{ret}}(Zk_j)=\alpha^{\mathrm T}b^{0}+\sum_{k\neq j}\beta_kZk_k\),
\[
  E\{(1-R)e\}=\frac{c_j}{\rho}\Big[\rho E\{(1-\pi_0)Sk_j\}-\alpha^{\mathrm T}E\{(1-R)b^{0}\}
  -\rho\sum_k\beta_kE\{(1-\pi_0)Sk_k\}\Big].
\] The first and third terms carry \(\rho\). For the second, the normal
equations defining \(P_{\mathrm{ret}}\) have cross-blocks
\(E(b^{0}Zk^{\mathrm T})=\rho E(b^{0}Sk^{\mathrm T})\) and right-hand
side \(E(b^{0}Zk_j)=\rho E(b^{0}Sk_j)\), both \(O(\rho)\), whence
\(\alpha=O(\rho)\). If in addition \(Z\) is independent of \(X\) then
\(E\{Zk_j(X)b^{0}(X)\}=E\{E(Z\mid X)k_jb^{0}\}=0\) and \(\alpha\)
vanishes identically.
\end{proof}

\begin{longtable}[]{@{}
  >{\raggedleft\arraybackslash}p{(\linewidth - 10\tabcolsep) * \real{0.0449}}
  >{\raggedleft\arraybackslash}p{(\linewidth - 10\tabcolsep) * \real{0.1923}}
  >{\raggedleft\arraybackslash}p{(\linewidth - 10\tabcolsep) * \real{0.1923}}
  >{\raggedleft\arraybackslash}p{(\linewidth - 10\tabcolsep) * \real{0.1795}}
  >{\raggedleft\arraybackslash}p{(\linewidth - 10\tabcolsep) * \real{0.2115}}
  >{\raggedleft\arraybackslash}p{(\linewidth - 10\tabcolsep) * \real{0.1795}}@{}}

\caption{\label{tbl-sieve}The projection coefficients on the non-\(Z\)
part of the retained sieve, sieve \({1,X,X^2,Z}\) with \(XZ\) omitted.
Left: \(Z\) independent of \(X\), where \(\alpha\) vanishes. Right:
\(S=(\varepsilon+X/2)/\sqrt{5/4}\), where \(\alpha=O(\rho)\).}

\tabularnewline

\toprule\noalign{}
\begin{minipage}[b]{\linewidth}\raggedleft
\(\rho\)
\end{minipage} & \begin{minipage}[b]{\linewidth}\raggedleft
\(\lVert\alpha\rVert\) (indep.)
\end{minipage} & \begin{minipage}[b]{\linewidth}\raggedleft
\(\lVert e^{*}\rVert\) (indep.)
\end{minipage} & \begin{minipage}[b]{\linewidth}\raggedleft
\(\lVert\alpha\rVert\) (dep.)
\end{minipage} & \begin{minipage}[b]{\linewidth}\raggedleft
\(\lVert\alpha\rVert/\rho\) (dep.)
\end{minipage} & \begin{minipage}[b]{\linewidth}\raggedleft
\(\lVert e^{*}\rVert\) (dep.)
\end{minipage} \\
\midrule\noalign{}
\endhead
\bottomrule\noalign{}
\endlastfoot
0.05 & 0.0040 & 1.002 & 0.0262 & 0.524 & 1.002 \\
0.10 & 0.0019 & 0.998 & 0.0462 & 0.462 & 0.999 \\
0.20 & 0.0016 & 0.998 & 0.0917 & 0.458 & 0.996 \\
0.35 & 0.0025 & 1.002 & 0.1586 & 0.453 & 0.989 \\
0.50 & 0.0016 & 1.003 & 0.2223 & 0.445 & 0.974 \\
0.80 & 0.0024 & 1.000 & 0.3555 & 0.444 & 0.935 \\

\end{longtable}

In Table~\ref{tbl-sieve} the prototype has \(Z\) independent of \(X\),
which is the degenerate case: \(\alpha\) is zero to Monte-Carlo
precision at every \(\rho\). Only when \(Z\) and \(X\) are dependent
does the claim \(\alpha=O(\rho)\) have content, and the ratio
\(\|\alpha\|/\rho\) is then constant to within 18\% across a sixteenfold
range of \(\rho\).

\section{Data preparation}\label{sec-data-prep}

The credit data comprise two files, of accepted and rejected
applications. Three conventions of the original analysis cannot be
inferred from the published text and are required to reproduce its
tables.

\begin{enumerate}
\def\labelenumi{(\roman{enumi})}
\item
  Each file carries one trailing all-empty record. Read under the
  conventions of the software then in use, this created an additional
  factor level sorting ahead of all others, which survives subsetting;
  the month-conversion tables of the original code are indexed by that
  shifted level index.
\item
  Durations are recoded to the lower bound of their interval plus six
  months (plus three for the two six-month intervals).
\item
  Age is recoded by \(\text{level index}+16\) followed by the
  replacement of the value \(18\) by \(65\). This is coherent only if
  the level \(\ge65\) sorts ahead of the numeric levels, which it did in
  the original locale; under that ordering every frequency matches its
  own label, and under any other it does not.
\end{enumerate}

Two errors in the original preparation are recovered by the reproduction
and reproduced deliberately so that the published tables can be
reconciled.

\emph{Age.} The recoding in (iii) interchanges the two extreme
categories: \(\ge65\) becomes \(17\) and \(18\)--\(20\) becomes \(65\),
while every other stated age is reduced by one year. Correcting it moves
the intercept and the contrast for pensioners by more than \(0.05\) and
no other coefficient; it does not affect any diagnostic reported.

\emph{Working duration.} One interval occurs only in the accepted file,
and the original code indexed the month-conversion table by factor
levels computed \emph{within each file}. Rejected applicants from that
level onward therefore received the duration of the next level; the 402
rejected applicants whose intended tenure is 9 months were coded as 558
months instead. Reproducing the per-file indexing is necessary to
recover the published coefficient of -0.002 for working duration; pooled
indexing returns +0.0014.

On this basis all 41 of the 41 coefficients of the published acceptance
model, their standard errors and \(p\)-values reproduce to three decimal
places. One printed value does not: the published table gives \(+0.495\)
for the contrast of three or more children, whereas magnitude, standard
error and \(p\)-value all match a coefficient of \(-0.495\), which the
monotone pattern across the categories supports; the later table of the
same source prints \(-0.496\).

\section{Computation}\label{computation}

Quantities involving \(\Lambda\) are evaluated on the logarithmic scale
by Gauss--Hermite quadrature, since Monte Carlo evaluation fails in
exactly the regime under study. The semiparametric likelihood is
maximized in its saddle form: with \(p_i=1/(nz_i)\) and
\(z_i=1+\lambda^{\mathrm T}g_i\), the inner problem in \(\lambda\) is
concave and is solved by damped Newton iteration, and only the outer
problem in the response coefficients is non-convex. Solving the full
score system directly is not advisable: its Jacobian has condition
number \(5\times10^{9}\) after column equilibration on the application
data---\(7.0\times10^{14}\) before it---and the iteration used in the
original analysis tests the size of its step rather than the size of the
score, so that it cannot detect failure to converge.

\end{document}